\documentclass[envcountersame,runningheads]{llncs}

\usepackage{hyperref}
\hypersetup{colorlinks=true,citecolor=blue,urlcolor=blue,linkcolor=blue}
\usepackage{color}

\usepackage{booktabs}
\usepackage[utf8]{inputenc}
\usepackage[T1]{fontenc}
\usepackage{amssymb}
\usepackage{amsmath}
\usepackage{cleveref}
\Crefname{figure}{Fig.}{Figs.}
\usepackage{scalerel}
\usepackage{mathtools}
\usepackage{tikz}
\usetikzlibrary{shapes,positioning,arrows}
\usepackage[misc]{ifsym}
\usepackage[noadjust]{cite}

\makeatletter
\newcommand{\superimpose}[2]{{%
    \ooalign{%
      \hfil$\m@th#1\@firstoftwo#2$\hfil\cr
      \hfil$\m@th#1\@secondoftwo#2$\hfil\cr
    }%
  }}
\makeatother

\newcommand{\m}[1]{\mathsf{#1}}
\newcommand{\mr}[1]{\mathrel{#1}}
\newcommand{\smallparallel}{\mathchoice%
  {\raisebox{1pt}{\scaleobj{.6}{\parallel}}} 
  {\raisebox{1pt}{\scaleobj{.6}{\parallel}}} 
  {\raisebox{.8pt}{\scaleobj{.45}{\parallel}}} 
  {\raisebox{.5pt}{\scaleobj{.35}{\parallel}}} 
}
\newcommand{\pto}{\mr{\mathpalette\superimpose{%
 {\smallparallel\kern.1em}{\to}}}}
\newcommand{\rpto}{\mr{\rotatebox[origin=c]{-180}{%
 $\mathpalette\superimpose{{\smallparallel\kern.1em}{\to}}$}}}
\newcommand{\rptoa}[1][]{\mr{\vphantom{\pto}^{#1}%
 \rotatebox[origin=c]{-180}{$\mathpalette\superimpose{%
 {\smallparallel\kern.1em}{\to}}$}}}
\newcommand{\smallcirc}{\mathchoice%
  {\raisebox{.5pt}{\scaleobj{.8}{\circ}}} 
  {\raisebox{.5pt}{\scaleobj{.8}{\circ}}} 
  {\raisebox{.8pt}{\scaleobj{.45}{\circ}}} 
  {\raisebox{.5pt}{\scaleobj{.35}{\circ}}} 
}
\newcommand{\mto}{\mr{\mathpalette\superimpose{%
 {\smallcirc\kern.1em}{\to}}}}
\newcommand{\union}{\mathop{\cup}}
\newcommand{\h}[1][.2]{\hspace{#1mm}}
\newcommand{\iN}{\h\in\h}
\newcommand{\geqslanT}{\geqslant\h}
\newcommand{\RB}[2][1]{\raisebox{#1mm}{#2}}
\newcommand{\MC}[2][0]{\makebox[#1mm]{#2}}
\newcommand{\ML}[2][0]{\makebox[#1mm][l]{#2}}

\newcommand{\xR}{\mathcal{R}}
\newcommand{\xRca}{\xR_\m{ca}}
\newcommand{\xRrc}{\xR_\m{rc}}

\newcommand{\xV}{\mathcal{V}}
\newcommand{\xF}{\mathcal{F}}
\newcommand{\xT}{\mathcal{T}}

\newcommand{\xJ}{\mathcal{J}}
\newcommand{\EC}{\mathcal{EC}}
\newcommand{\Val}{\xV\m{al}}
\newcommand{\Var}{\xV\m{ar}}
\newcommand{\LVar}{\mathcal{L}\Var}
\newcommand{\EVar}{\mathcal{E}\Var}

\newcommand{\TVar}{\mathcal{T}\Var}
\newcommand{\Dom}{\mathcal{D}\m{om}}

\newcommand{\Pos}{\mathcal{P}\m{os}}

\newcommand{\FPos}{\Pos_\xF}
\newcommand{\inter}[1]{[\![{#1}]\!]}
\newcommand{\seq}[2][n]{{#2_1},\dots,{#2_{#1}}}

\newcommand{\NN}{\mathbb{N}}
\newcommand{\ZZ}{\mathbb{Z}}
\newcommand{\mte}{\m{te}}
\newcommand{\mth}{\m{th}}
\newcommand{\xFTe}{\xF_{\m{te}}}
\newcommand{\xFTh}{\xF_{\m{th}}}
\newcommand{\SET}[1]{\{\h#1\h\}}
\newcommand{\set}[1][n]{\{\h1,\dots,#1\h\}}
\newcommand{\ov}[1]{\overline{#1}}
\newcommand{\CO}[1]{[\h#1\h]}
\newcommand{\R}{\rightarrow}
\newcommand{\Ra}[1][]{\R^{#1}}
\newcommand{\Rb}[1][]{\R_{#1}}
\newcommand{\Rab}[2][]{\R_{#1}^{#2}}
\newcommand{\RbR}{\Rb[\xR]}
\newcommand{\RboR}{\Rb[\smash{\ov{\xR}}]}
\newcommand{\xRa}[1][]{\xrightarrow{#1}}
\renewcommand{\L}{\leftarrow}
\newcommand{\La}[1][]{\mr{\vphantom{\R}^{#1}{\L}}}

\newcommand{\CP}{\m{CP}}
\newcommand{\PCP}{\m{PCP}}
\newcommand{\CCP}{\m{CCP}}
\newcommand{\CPCP}{\m{CPCP}}
\newcommand{\overlap}[3][p]{\langle #2, #1, #3 \rangle}
\newcommand{\crr}[3]{#1 \R #2~\CO{#3}}
\newcommand{\CRR}{\crr{\ell}{r}{\varphi}}
\newcommand{\ccp}[3]{#1 \approx #2~\CO{#3}}
\newcommand{\Rs}{\stackrel{\smash{\RB[-.5]{\tiny $\sim$~}}}{\R}}
\newcommand{\sRb}[1]{\Rs_{#1}}

\newcommand{\sRbR}{\sRb{\xR}}
\newcommand{\sRab}[2][]{\Rs_{#1}^{#2}}
\newcommand{\MS}[1][]{\xRa[\smash{%
 \RB[-2]{\MC[2.5]{$\stackrel{#1}{\circ}$}}}]}
\newcommand{\sMS}{\stackrel{\smash{\RB[-1.2]{\tiny $\sim$\:}}}{\MS}}
\newcommand{\sMSab}[2][]{\sMS_{#1}^{#2}}
\newcommand{\spto}{\stackrel{\smash{\RB[-1.0]{\tiny $\sim$\:}}}{\pto}}
\newcommand{\sptoab}[2][]{\mr{\smash{\spto}_{#1}^{#2}}}
\newcommand{\crest}{\textsf{crest}}

\title{Confluence of Logically Constrained Rewrite Systems Revisited%
\thanks{This research is funded by the Austrian Science Fund (FWF) project
I5943.}}
\titlerunning{Confluence of LCTRSs Revisited}

\author{Jonas Sch\"opf%
\textsuperscript{(\href{mailto:jonas.schoepf@uibk.ac.at}{\Letter})}%
\orcidID{0000-0001-5908-8519} \and
Fabian Mitterwallner\orcidID{0000-0001-5992-9517} \and
Aart Middeldorp\orcidID{0000-0001-7366-8464}}
\authorrunning{Sch\"opf, Mitterwallner, Middeldorp}

\institute{Department of Computer Science, University of Innsbruck,
Innsbruck, Austria
\email{\{jonas.schoepf,fabian.mitterwallner,aart.middeldorp\}@uibk.ac.at}}

\begin{document}

\maketitle

\begin{abstract}
We show that (local) confluence of terminating logically
constrained rewrite systems is undecidable, even when the underlying
theory is decidable.
Several confluence criteria for logically constrained rewrite systems
are known. These were obtained by replaying existing proofs for plain
term rewrite systems in a constrained setting, involving a non-trivial
effort. We present a simple transformation from logically constrained
rewrite systems to term rewrite systems such that critical pairs of the
latter correspond to constrained critical pairs of the former. The
usefulness of the transformation is illustrated by lifting the advanced
confluence results based on (almost) development closed critical pairs as
well as on parallel critical pairs to the constrained setting.
\end{abstract}

\section{Introduction}

Logically constrained rewrite systems (LCTRSs)~\cite{KN13} are a natural
extension of plain term rewrite systems (TRSs) with native support for
constraints that are handled by SMT solvers. The latter makes LCTRSs
suitable for program analysis~\cite{CL18,CLB23,FKN17,WM18}. 
In this paper we are concerned with confluence techniques for LCTRSs.
Numerous techniques exist to (dis)prove confluence of TRSs.
For LCTRSs much less is known. Kop and Nishida~\cite{KN13}
established
(weak) orthogonality as sufficient confluence criteria for LCTRSs.
Joinability of critical pairs for terminating systems is implicit in
\cite{WM18}. Very recently, strong closedness for linear LCTRSs
and (almost) parallel closedness for left-linear LCTRSs were
established \cite{SM23}. The proofs of these results were obtained by
\emph{replaying} existing proofs for TRSs in a constrained setting,
involving a non-trivial effort. For more advanced confluence criteria,
this is not feasible.

In particular, the conclusion in \cite{KN13} that LCTRSs ``are
\emph{flexible}: common analysis techniques for term rewriting extend to
LCTRSs without much effort'' is not accurate. On the contrary, in
\Cref{sec:undecidability} we show that (local) confluence of terminating
LCTRSs is undecidable, even for a decidable fragment of the theory of
integers.

In \Cref{sec:transformation} we present a simple transformation from
LCTRSs to TRSs which allows us to relate results for the latter to the
former. We use the transformation to extend two advanced confluence
criteria based on (parallel) critical pairs from TRSs to LCTRSs:
In \Cref{sec:dev-closed-ccps} we prove that
(almost) development closed left-linear LCTRSs are confluent by
\emph{reusing} the corresponding result for TRSs obtained by van
Oostrom~\cite{vO97} and in \Cref{sec:pcps} we lift the result of
Toyama~\cite{T81} based on parallel critical pairs from TRSs to LCTRSs.
Both results are employed in state-of-the-art confluence provers for
TRSs (\textsf{ACP}~\cite{AYT09}, \textsf{CSI}~\cite{NFM17},
\textsf{Hakusan}~\cite{SH22}) and have only recently been formally
verified in the Isabelle proof assistant~\cite{HKST24,KM23a,KM23b}.

For the LCTRS extension of the result of Toyama~\cite{T81} we
observed a subtle problem in the definition of the equivalence relation
on constrained terms, which goes back to \cite{KN13} and has been
used in subsequent work on LCTRSs~\cite{FKN17,SM23,WM18}. We briefly
discuss the issue at the end of the next section, after recalling
basic notions for LCTRSs.
The results in \Cref{sec:transformation} and
\Cref{sec:dev-closed-ccps} were first announced in~\cite{MSM23}.

\section{Preliminaries}
\label{sec:prelims}

We assume familiarity with the basic notions of term rewriting. In
this section we recall a few key notions
for LCTRSs. For more background information we refer
to~\cite{KN13,SM23,WM18}. We assume a many-sorted signature
$\xF = \xFTe \cup \xFTh$ with a \textsf{te}rm and \textsf{th}eory part.
For every sort $\iota$ in $\xFTh$ we have a
non-empty set $\Val_\iota \subseteq \xFTh$ of value
symbols, such that all $c \in \Val_\iota$ are constants of sort $\iota$.
We demand $\xFTe \cap \xFTh \subseteq \Val$ where
$\Val = \bigcup_\iota \Val_\iota$.
In the case of integers this results
in an infinite signature with $\mathbb{Z} \subseteq \Val \subseteq \xFTh$.
A term in $\xT(\xFTh,\xV)$ is called a \emph{logical} term.
Ground logical terms are mapped to values by an
interpretation $\xJ$:
$\inter{f(\seq t)} = f_\xJ(\inter{t_1},\dots,\inter{t_n})$.
We assume a bijection between value symbols and elements in the domain
of $\xJ$, e.g., for integers: $\inter{\m{0}} = 0$,
$\inter{\m{-1}} = -1$, $\inter{\m{1}} = 1$ and so on.
Logical terms of sort $\m{bool}$ are called \emph{constraints}.
A constraint $\varphi$ is \emph{valid} if
$\inter{\varphi\gamma} = \top$ for all substitutions $\gamma$ such that
$\gamma(x) \in \Val$ for all $x \in \Var(\varphi)$.
A \emph{constrained rewrite rule} is a triple $\CRR$ where
$\ell, r \in \xT(\xF,\xV)$ are terms of the same sort such that
$\m{root}(\ell) \in \xFTe \setminus \xFTh$ and $\varphi$ is a
constraint.
We denote the set $\Var(\varphi) \cup (\Var(r) \setminus \Var(\ell))$ of
\emph{logical} variables in $\CRR$ by $\LVar(\CRR)$.
A constrained rewrite rule is left-linear (right-linear) if non-logical
variables in the left-hand side (right-hand side) occur at most once. If
a rule is left-linear and right-linear then it is called linear.
An LCTRS is a set of constrained rewrite rules.

A substitution $\sigma$ is said to \emph{respect} a rule $\CRR$, denoted
by $\sigma \vDash \CRR$, if
$\Dom(\sigma) \subseteq \Var(\ell) \cup \Var(r) \cup \Var(\varphi)$,
$\sigma(x) \in \Val$ for all $x \in \LVar(\CRR)$, and $\inter{\varphi\sigma}
= \top$. Moreover, a constraint $\varphi$ is respected by $\sigma$,
denoted by $\sigma \vDash \varphi$, if $\sigma(x) \in \Val$ for all
$x \in \Var(\varphi)$ and $\inter{\varphi\sigma} = \top$. We call
$f(\seq{x}) \R y~\CO{y = f(\seq{x})}$ with a fresh variable $y$
and $f \in \xFTh \setminus \Val$ a \emph{calculation rule}.
Calculation rules are not part of the rules of an LCTRS $\xR$.
The set of
all calculation rules induced by the signature $\xFTh$ of an LCTRS
$\xR$ is denoted by $\xRca$ and we abbreviate $\xR \cup \xRca$ to $\xRrc$.
An LCTRS is called linear (left-linear, right-linear) if all its
rules in $\xR$ are linear (left-linear, right-linear).
A rewrite step $s \RbR t$ satisfies $s|_p = \ell\sigma$ and
$t = s[r\sigma]_p$
for some position $p$, constrained rewrite rule $\CRR$ in $\xRrc$, and
substitution $\sigma$ such that $\sigma \vDash \CRR$.
We drop the subscript $\xR$ from $\RbR$ when no confusion arises.
An LCTRS $\xR$ is confluent if there exists a term $v$ with
$t \Ra[*] v \La[*] u$ whenever $t \La[*] s \Ra[*] u$, for all terms $s$,
$t$ and $u$. For confluence analysis we need to rewrite constrained terms.

A \emph{constrained term} is a pair $s~\CO{\varphi}$ consisting
of a term $s$ and a constraint $\varphi$. Two constrained terms
$s~\CO{\varphi}$ and $t~\CO{\psi}$ are \emph{equivalent},
denoted by $s~\CO{\varphi} \sim t~\CO{\psi}$, if for every substitution
$\gamma \vDash \varphi$ with $\Dom(\gamma) = \Var(\varphi)$
there is some substitution $\delta \vDash \psi$ with
$\Dom(\delta) = \Var(\psi)$ such that $s\gamma = t\delta$, and vice versa.
Let $s~\CO{\varphi}$ be a constrained term.
If $s|_p = \ell\sigma$ for some constrained rewrite rule
$\rho\colon \crr{\ell}{r}{\psi} \in \xRrc$, position $p$, and
substitution $\sigma$ such that $\sigma(x) \in \Val \cup \Var(\varphi)$
for all $x \in \LVar(\rho)$, $\varphi$ is satisfiable and
$\varphi \Rightarrow \psi\sigma$ is valid then
$s~\CO{\varphi} \RbR s[r\sigma]_p~\CO{\varphi}$.
The rewrite relation $\sRbR$ on constrained terms is defined as
$\sim \cdot \RbR \cdot \sim$ and
$s~\CO{\varphi} \sRb{p} t~\CO{\psi}$
indicates that the rewrite step in $\sRbR$ takes place at position
$p$. Similarly, we write
$s~\CO{\varphi} \sRb{\geqslanT p} t~\CO{\psi}$
if the position in the rewrite step is below position $p$.
Note that in our definition of $\RbR$ the constraint is not
modified. This equals \cite[Definition~2.15]{FKN17}, but is different from
\cite{KN13,SM23} where calculation steps
$s[f(\seq{v})]_p~\CO{\varphi} \R s[v]_p~\CO{\varphi \land v = f(\seq{v})}$
modify the constraint.
However, the relation $\Rs$ can simulate the relation $\RbR$
from~\cite{KN13,SM23} as exemplified below.

\begin{example}
Consider the constrained term $x + \m{1}~\CO{x > \m{3}}$. Calculation
steps as defined in \cite{KN13,SM23} permit
$x + \m{1}~\CO{x > \m{3}} \R z~\CO{z = x + \m{1} \land x > \m{3}}$.
In our setting, an initial equivalence step is required to introduce the
fresh variable $z$ and the corresponding
assignment needed to perform a calculation:
$x + \m{1}~\CO{x > \m{3}} \sim x + \m{1}~\CO{z = x + \m{1} \land x > \m{3}}
\R z~\CO{z = x + \m{1} \land x > \m{3}}$.
\end{example}

Our treatment allows for a much simpler definition of parallel and
multi-step rewriting since we do not have to merge different constraints.

\subsection*{Equivalence on Constrained Terms}

The equivalence on constrained terms $\sim$ used in this paper
also differs from the equivalence relation used in~\cite{KN13,SM23},
which we will denote by $\sim'$.
In $\sim'$ the domain of substitutions is not restricted, i.e., 
$s~\CO{\varphi} \sim' t~\CO{\psi}$ if and only if for all substitutions
$\gamma \vDash \varphi$ there exists a substitution $\delta$ where
$\delta \vDash \psi$ and $s\gamma = t\delta$.
Intuitively, constrained terms are equivalent with respect to
$\sim'$ if their sets of ``allowed'' instances are equivalent, while for
$\sim$ we only instantiate variables appearing in the constraints and
therefore representing some value.
We have ${\sim} \subsetneq {\sim'}$. This can be seen as follows.
First of all, any substitution $\gamma$ with $\gamma \vDash \varphi$
can be split into $\gamma_1$ and $\gamma_2$ such that
$\gamma = \gamma_1 \cup \gamma_2 = {} \gamma_1\gamma_2$ with
$\Dom(\gamma_1) = \Var(\varphi)$ and $\gamma_1 \vDash \varphi$.
From $s~\CO{\varphi} \sim t~\CO{\psi}$ we obtain a substitution
$\delta_1$ where $\Dom(\delta_1) = \Var(\psi)$, $\delta_1 \vDash \psi$ and
$s\gamma_1 = t\delta_1$.
Hence also $s\gamma = s\gamma_1\gamma_2 = t\delta_1\gamma_2 = t\delta$
for $\delta = \delta_1\gamma_2$,
which implies $s~\CO{\varphi} \sim' t~\CO{\psi}$.
However, ${\sim'} \subseteq {\sim}$ does not hold since
$x~\CO{\m{true}} \sim' y~\CO{\m{true}}$ and
$x~\CO{\m{true}} \not\sim y~\CO{\m{true}}$.

The change is necessary, since we have to differentiate (non-logical)
variables in constrained terms from one another, to keep track of them
through rewrite sequences. Take the (LC)TRS $\xR$ consisting of the rule
$\m{f}(x,y) \to x$. When rewriting unconstrained terms we have
$\m{f}(x,y) \RbR x$ and $\m{f}(x,y) \not\RbR y$. When rewriting on
constrained terms with respect to $\sim'$, however, we have
$\m{f}(x,y)~\CO{\m{true}} \sim' \cdot \to \cdot \sim' x~\CO{\m{true}}$
and
$\m{f}(x,y)~\CO{\m{true}} \sim' \cdot \to \cdot \sim' y~\CO{\m{true}}$,
losing any information connecting the resulting variable to the initial
term. This is especially problematic in our analysis of parallel critical
pairs in \Cref{sec:pcps}, where keeping track of variables through
rewrite sequences is essential.
Note that $\m{f}(x,y)~\CO{\m{true}} \Rs x~\CO{\m{true}}$ but not
$\m{f}(x,y)~\CO{\m{true}} \Rs y~\CO{\m{true}}$.

\section{Undecidability}
\label{sec:undecidability}

Confluence is a decidable property of finite terminating TRSs,
a celebrated result of Knuth and Bendix~\cite{KB70} which forms the
basis of completion. For LCTRSs matters are more complicated.

\begin{theorem}
\label{thm:undecidable}
Local confluence is undecidable for terminating \textup{LCTRS}s.
\end{theorem}

\begin{proof}
We use a reduction from PCP~\cite{P46}. Let
$P = \SET{(\alpha_1,\beta_1),\dots,(\alpha_N,\beta_N)}$ with
$\seq[N]{\alpha}, \seq[N]{\beta} \in \SET{0,1}^+$ be an instance of PCP,
where we assume that $\alpha_i \neq \beta_i$ for at least one
$i \in \set[N]$. This entails no loss of generality, since instances that
violate this assumption are trivially solvable. We encode candidate strings
over $\set[N]$ as natural numbers
where the empty string $\epsilon$ is represented by $[\epsilon] = 0$,
and a non-empty string $i_0 i_1 \cdots i_k$ is represented by
$[i_0 i_1 \cdots i_k] = N \cdot [i_1 \cdots i_k] + i_0$.
So $[i_0 i_1 \cdots i_k] = i_0 + i_1 \cdot N +
\cdots + i_k \cdot N^k$. For instance, assuming $N = 3$, the number $102$
encodes the candidate string $\underline{3313}$
since $102 = 3 \cdot 33 + \underline{3}$,
$33 = 3 \cdot 10 + \underline{3}$,
$10 = 3 \cdot 3 + \underline{1}$ and $3 = 3 \cdot 0 + \underline{3}$.
Conversely, the candidate string $\underline{112}$ is mapped to
$22 = \underline{1} + \underline{1} \cdot 3^1 + \underline{2} \cdot 3^2$.
It is not difficult to see that this results in a bijection between
$\NN$ and candidate strings, for each $N > 0$.

The LCTRS $\xR_P$ that we construct is defined over the theory
$\m{Ints}$, with theory symbols
$\xFTh = \SET{>,+,\cdot,=,\land} \cup \Val$
and values $\Val = \mathbb{B} \cup \mathbb{Z}$, with the additional sorts
$\m{PCP}$ and $\m{String}$ and the following term signature:
\begin{align*}
\m{e} : {} &\m{String} &
\m{0}, \m{1} : {} &\m{String} \to \m{String} \\
\m{start}, \top, \bot : {} & \m{PCP} &
\m{test} :  {} &\m{String} \times \m{String} \to \m{PCP} \\
\m{alpha}, \m{beta} : {} &\m{Int} \to \m{String}
\end{align*}
The LCTRS $\xR_P$ consists of the following rules:
{\allowdisplaybreaks
\begin{align*}
\m{start} &\to \mathrlap{\m{test}(\m{alpha}(n),\m{beta}(n)) \quad
\CO{n > 0}} \\
\m{test}(\m{e},\m{e}) &\to \top \\
\m{test}(\m{0}(x),\m{0}(y)) &\to \m{test}(x,y) &
\m{test}(\m{0}(x),\m{1}(y)) &\to \bot \\
\m{test}(\m{1}(x),\m{1}(y)) &\to \m{test}(x,y) &
\m{test}(\m{1}(x),\m{0}(y)) &\to \bot \\
\m{test}(\m{0}(x),\m{e}) &\to \bot &
\m{test}(\m{e},\m{0}(y)) &\to \bot \\
\m{test}(\m{1}(x),\m{e}) &\to \bot &
\m{test}(\m{e},\m{1}(y)) &\to \bot \\
\m{alpha}(0) &\to \m{e} &
\m{beta}(0) &\to \m{e} \\
\intertext{and, for all $i \in \set[N]$,}
\m{alpha}(n) &\to \ML[20]{$\alpha_i(\m{alpha}(m))$} \quad
\mathrlap{\CO{N \cdot m + i = n \land n > 0}} \\
\m{beta}(n) &\to \ML[20]{$\beta_i(\m{beta}(m))$} \quad
\mathrlap{\CO{N \cdot m + i = n \land n > 0}}
\end{align*}}
Here, for a string $\gamma \in \SET{0,1}^*$ and a term $t : \m{String}$,
$\gamma(t) : \m{String}$ is defined as
\[
\gamma(t) = \begin{cases}
t &\text{if $\gamma = \epsilon$} \\
\m{0}(\gamma'(t)) &\text{if $\gamma = 0\gamma'$} \\
\m{1}(\gamma'(t)) &\text{if $\gamma = 1\gamma'$}
\end{cases}
\]
Note that in the constraints $n$ and $m$ are variables, while $N$ and $i$
are values. Hence all constraints are in the decidable fragment of linear
integer arithmetic and the rewrite relation $\to_{\xR_P}$ is computable.

We claim that $\xR_P$ is locally confluent if and only if $P$ has no
solution. The LCTRS $\xR_P$ admits the constrained critical pair
\[
\m{test}(\m{alpha}(n),\m{beta}(n)) \approx
\m{test}(\m{alpha}(m),\m{beta}(m)) \quad \CO{n > 0 \,\land\, m > 0}
\]
with $n \neq m$.
The rules with left-hand sides $\m{alpha}(n)$ and $\m{beta}(n)$ give rise
to further constrained critical pairs but these are harmless since for all
$n, N > 0$ there are unique numbers $i$ and $m$ satisfying the
constraint $\CO{N \cdot m + i = n \land n > 0}$.

By construction of the rules for $\m{test}$,
$\m{test}(\m{alpha}(n),\m{beta}(n)) \to^* \top$ if $n$ represents a
solution of $P$ and $\m{test}(\m{alpha}(n),\m{beta}(n)) \to^* \bot$ if
$n$ does not represent a solution of $P$. Since we assume that $P$ is
non-trivial, the latter happens for some $n > 0$. Hence all instances of
the constrained critical pairs can only be joined if
$\m{test}(\m{alpha}(n),\m{beta}(n)) \to^* \bot$ for all $n > 0$.
Hence $\xR_P$ is locally confluent if and only if $P$ has no solution.

The LCTRS $\xR_P$ is terminating by the recursive path order~\cite{KN13}
with the precedence $\m{start} > \m{test} > \m{alpha} > \m{beta} > \m{1} >
\m{0} > \m{e} > \m{\top} > \m{\bot}$ and the well-founded order
$\sqsupset_{\m{Int}}$ on integers where
$x \sqsupset_{\m{Int}} y$ if and only if $x > y$ and
$x \geqslant 0$. The key observation is that the constraint
$\CO{N \cdot m + i = n \land n > 0}$ in the recursive rules for
$\m{alpha}$ and $\m{beta}$ ensure $n > m$ since
$N > 0$ and $i \geqslant 1$.
\qed
\end{proof}

A key difference between TRSs and LCTRSs leading to this undecidability
result can be seen in the first rule:
$\m{start} \to \m{test}(\m{alpha}(n),\m{beta}(n))\:\CO{n > 0}$.
Plain TRSs usually do not allow variables appearing only in the
right-hand side of a rule, as is the case for $n$ here, because then
termination never holds.
However, in LCTRSs such variables are useful,
since they can be used to model computations on arbitrary values
which are often used to represent user input in program analysis.
For $\xR_P$ this leads to infinitely many possible steps starting from
the term $\m{start}$ and in turn to infinitely many critical pairs,
breaking decidability.

\section{Transformation}
\label{sec:transformation}

In this section we present a simple transformation from LCTRSs to
possibly infinite TRSs, which exactly corresponds to the intuition
behind LCTRSs. This allows us to lift results on TRSs more easily to
LCTRSs than previously possible.

\begin{definition}
\label{def:transformation}
Given an \textup{LCTRS} $\xR$, the \textup{TRS} $\ov{\xR}$ consists of the
following rules: $\ell\tau \R r\tau$
for all $\rho\colon \CRR \in \xRrc$ with $\tau \vDash \rho$ and
$\Dom(\tau) = \LVar(\rho)$.
\end{definition}

Note that $\ov{\xR}$ typically consists of infinitely many rules.

\begin{lemma}
\label{lem:overline relation identity}
The rewrite relations of $\xR$ and $\ov{\xR}$ are the same.
Moreover ${\Rb[p,\xR]} = {\Rb[p,\ov{\xR}]}$ for all positions $p$.
\end{lemma}

\begin{proof}
We first show ${\Rb[p,\xR]} \subseteq {\Rb[p,\ov{\xR}]}$.
Assume $s \Rb[p,\xR] t$. We have
$s = s[\ell\sigma]_p \R s[r\sigma]_p = t$ for some
$\rho\colon \CRR \in \xRrc$ and $\sigma \vDash \rho$.
We split $\sigma$ into two substitutions
$\tau = \SET{x \mapsto \sigma(x) \mid x \in \LVar(\rho)}$ and
$\delta = \SET{x \mapsto \sigma(x) \mid x \in \Var(\ell) \setminus
\LVar(\rho)}$. From $\sigma \vDash \rho$ we infer $\tau \vDash \rho$
and thus $\tau(x) \in \Val$ for all $x \in \LVar(\rho)$. Hence
$\sigma = \tau \cup \delta = \tau\delta$. We have
$\ell\tau \R r \tau \in \ov{\xR}$.
Hence $s = s[\ell\tau\delta]_p \Rb[p,\ov{\xR}] s[r\tau\delta]_p = t$ as
desired. To show the reverse inclusion
${\Rb[p,\ov{\xR}]} \subseteq {\Rb[p,\xR]}$ we assume
$s \Rb[p,\ov{\xR}] t$. Otherwise
$s = s[\ell\mu\nu]_p \Rb[p,\ov{\xR}] s[r\mu\nu]_p$ for some rule
$\rho\colon \CRR \in \xR$ with $\mu \vDash \rho$. Let
$\sigma = \mu\nu$. Since $\mu(x) \in \Val$ for all
$x \in \LVar(\rho)$, we have $x\sigma = x\mu$ for all
$x \in \LVar(\rho)$. Hence $\sigma \vDash \rho$ and thus
$s = s[\ell\sigma]_p \Rb[p,\xR] s[r\sigma]_p = t$.
\qed
\end{proof}

Since $\RbR$ and $\RboR$ coincide, we drop the subscript in the sequel.
We write $\EVar(\CRR)$ for the set $\Var(r) \setminus (\Var(\ell) \cup
\Var(\varphi))$ of extra variables of a rule. In the computation of
constrained critical pairs these variables of the overlapping rules
would lose the property of being a logical variable without adding
trivial constraints. Given a constrained rewrite rule $\rho$, we write
$\EC_\rho$ for $\bigwedge \SET{x = x \mid x \in \EVar(\rho)}$.
The set of positions in a term $s$ is denoted by $\Pos(s)$.
We write $\epsilon$ for the root position and $\FPos(s)$ for
the set of positions of function symbols in $s$.

\begin{definition}
\label{def:ccp}
An \emph{overlap} of an \textup{LCTRS} $\xR$ is a triple
$\overlap{\rho_1}{\rho_2}$
with rules $\rho_1\colon \crr{\ell_1}{r_1}{\varphi_1}$ and
$\rho_2\colon \crr{\ell_2}{r_2}{\varphi_2}$, satisfying the following
conditions: (1) $\rho_1$ and $\rho_2$ are variable-disjoint variants of
rewrite rules in $\xRrc$, (2) $p \in \FPos(\ell_2)$, (3) $\ell_1$ and
$\ell_2|_p$ unify with mgu $\sigma$ such that
$\sigma(x) \in \Val \cup \xV$ for all
$x \in \LVar(\rho_1) \cup \LVar(\rho_2)$, (4)
$\varphi_1\sigma \land \varphi_2\sigma$ is satisfiable, and (5) if
$p = \epsilon$ then $\rho_1$ and $\rho_2$ are not variants, or
$\Var(r_1) \nsubseteq \Var(\ell_1)$. In this case we call
$\ell_2\sigma[r_1\sigma]_p \approx r_2\sigma~
\CO{\varphi_1\sigma \land \varphi_2\sigma \land \psi\sigma}$
a \emph{constrained critical pair} obtained from the overlap
$\overlap{\rho_1}{\rho_2}$. Here
$\psi = \EC_{\rho_1} \land \EC_{\rho_2}$.
The peak
$\ell_2\sigma[r_1\sigma]_p~\CO{\Phi} \L \ell_2\sigma~\CO{\Phi}
\Rb[\epsilon] r_2\sigma~\CO{\Phi}$
with $\Phi = (\varphi_1 \land \varphi_2 \land \psi)\sigma$,
from which the constrained critical pair originates, is called a
\emph{constrained critical peak}. The set of all constrained critical
pairs of $\xR$ is denoted by $\CCP(\xR)$.
A constrained critical pair $s \approx t~\CO{\varphi}$ is \emph{trivial}
if $s\sigma = t\sigma$ for every substitution $\sigma$ with
$\sigma \vDash \varphi$.
\end{definition}

A key ingredient of our approach is to relate critical pairs of the
transformed TRS to constrained critical pairs of the original LCTRS.

\begin{theorem}
\label{thm:CP}
For every critical pair $s \approx t$ of $\ov{\xR}$
there exists a constrained critical pair $\ccp{s'}{t'}{\varphi'}$ of
$\xR$ and a substitution $\gamma$ such that $s = s'\gamma$, $t = t'\gamma$
and $\gamma \vDash \varphi'$.
\end{theorem}

\begin{proof}
Let $s \approx t$ be a critical pair of $\ov{\xR}$, originating from 
the critical peak $\ell_2\mu\sigma[r_1\nu\sigma]_p \L \ell_2\mu\sigma =
\ell_2\mu\sigma[\ell_1\nu\sigma]_p \R r_2\mu\sigma$ with variants
$\rho_1\colon \crr{\ell_1}{r_1}{\varphi_1}$ and
$\rho_2\colon \crr{\ell_2}{r_2}{\varphi_2}$
of rules in $\xRrc$ without shared variables. Let
$\psi_i = \EC_{\rho_i}$ for $i \in \SET{1,2}$.
Furthermore we have $\Dom(\nu) = \LVar(\rho_1)$,
$\Dom(\mu) = \LVar(\rho_2)$,
$\nu \vDash \varphi_1 \land \psi_1$, $\mu \vDash \varphi_2 \land \psi_2$,
$p \in \FPos(\ell_2\mu)$, and $\sigma$ is an mgu of
$\ell_2\mu|_p$ and $\ell_1\nu$. Moreover, if $p = \epsilon$ then
$\ell_1\nu \R r_1\nu$ and $\ell_2\mu \R r_2\mu$ are not variants.
Define $\tau = \nu \uplus \mu$. We have
$\Dom(\tau) = \LVar(\rho_1) \cup \LVar(\rho_2)$. Let
$\varphi = \varphi_1 \land \varphi_2 \land \psi_1 \land \psi_2$. Clearly,
$\ell_1\tau = \ell_1\nu$, $r_1\tau = r_1\nu$, $\ell_2\tau = \ell_2\mu$,
$r_2\tau = r_2\mu$ and $\tau \vDash \varphi$.
Hence the given peak can be written as $\ell_2\tau\sigma[r_1\tau\sigma]_p
\L \ell_2\tau\sigma = \ell_2\tau\sigma[\ell_1\tau\sigma]_p \R
r_2\tau\sigma$ and $\tau \vDash \varphi$.
Since $\ell_2|_p\tau\sigma = \ell_1\tau\sigma$ there exists an mgu
$\delta$ of $\ell_2|_p$ and $\ell_1$, and a substitution $\gamma$ such
that $\delta\gamma = \tau\sigma$. Let $s' = \ell_2\delta[r_1\delta]_p$ and
$t' = r_2\delta$. We claim that $\overlap{\rho_1}{\rho_2}$ is an
overlap of $\xR$, resulting in the constrained critical pair
$\ccp{s'}{t'}{\varphi\delta}$.
Condition (1) of \Cref{def:ccp} is trivially satisfied.
For condition (2) we need to show $p \in \FPos(\ell_2)$. This follows
from $p \in \FPos(\ell_2\mu)$, $\mu(x) \in \Val$ for every
$x \in \Dom(\mu)$, and $\m{root}(\ell_2\mu|_p) = \m{root}(\ell_1\nu) \in
\xF \setminus \Val$. For condition (3) it remains to show that
$\delta(x) \in \Val \cup \xV$ for all
$x \in \LVar(\rho_1) \cup \LVar(\rho_2)$.
Suppose to the contrary that $\m{root}(\delta(x)) \in \xF \setminus \Val$
for some $x \in \LVar(\rho_1) \union \LVar(\rho_2)$. Then
$\m{root}(\delta(x)) = \m{root}(\gamma(\delta(x))) =
\m{root}(\sigma(\tau(x))) \in \mathcal{F} \setminus \Val$,
which contradicts $\tau \vDash \varphi$.
Condition (4) follows from the identity $\delta\gamma = \tau\sigma$
together with $\tau \vDash \varphi$ which imply
$\delta\gamma \vDash \varphi$
and thus $\varphi\delta$ is satisfiable.
Hence also $\varphi_1\delta \land \varphi_2\delta$ is satisfiable.
It remains to show condition (5), so let $p = \epsilon$ and further assume
that $\rho_1$ and $\rho_2$ are variants. So there exists a variable
renaming $\pi$ such that $\rho_1\pi = \rho_2$. In particular,
$\ell_1\pi = \ell_2$ and $r_1\pi = r_2$.
Let $x \in \Var(\ell_1)$. If
$x \in \LVar(\rho_1) = \Dom(\nu)$ then $\tau(x) = \nu(x) \in \Val$.
Moreover, $\pi(x) \in \LVar(\rho_2) = \Dom(\mu)$ and thus
$\tau(\pi(x)) = \mu(\pi(x)) \in \Val$. Since $\ell_1\tau$ and
$\ell_2\tau$ are unifiable, $\pi(\tau(x)) = \tau(x) = \tau(\pi(x))$.
If $x \notin \LVar(\rho_1)$ then $\tau(x) = x$,
$\pi(x) \notin \LVar(\rho_2)$
and similarly $\tau(\pi(x)) = \pi(x) = \pi(\tau(x))$.
All in all, $\ell_1\tau\pi = \ell_1\pi\tau = \ell_2\tau$.
Now, if $\Var(r_1) \subseteq \Var(\ell_1)$ then we obtain
$r_1\tau\pi = r_1\pi\tau = r_2\tau$, contradicting the fact that
$\ell_1\nu \R r_1\nu$ and $\ell_2\mu \R r_2\mu$ are not variants.
We conclude that
$\ccp{s'}{t'}{\varphi\delta}$ is a constrained critical pair of $\xR$.
So we can take $\varphi' = \varphi\delta$.
Clearly, $s = s'\gamma$ and $t = t'\gamma$. Moreover,
$\gamma \vDash \varphi'$ since
$\varphi'\gamma = \varphi\tau\sigma = \varphi\tau$ and
$\tau \vDash \varphi$.
\qed
\end{proof}

The converse does not hold in general.

\begin{example}
\label{exa:weakly-orthogonal}
Consider the LCTRS $\xR$
consisting of the single rule $\crr{\m{a}}{x}{x = 0}$ where
the variable $x$ ranges over the integers. Since $x$ appears on the
right-hand side but not the left, we obtain a constrained critical pair
$\ccp{x}{x'}{x = 0 \land x' = 0}$. Since the constraint uniquely
determines the values of $x$ and $x'$, the TRS $\ov{\xR}$ consists
of the single rule $\m{a} \to 0$. Obviously $\ov{\xR}$ has no critical
pairs.
\end{example}

The above example also shows that orthogonality of $\ov{\xR}$ does
not imply orthogonality of $\xR$.
However, the counterexample relies somewhat on a technicality in
condition (5) of \Cref{def:ccp}. It only occurs when the two rules
$\crr{\ell_1}{r_1}{\varphi_1}$
and $\crr{\ell_2}{r_2}{\varphi_2}$
involved in the critical pair
overlap at the root and have instances
$\ell_1\tau_1 \R r_1\tau_1$ and $\ell_2\tau_2 \R r_2\tau_2$
in $\ov{\xR}$ which are variants of each other.
By dealing with such cases separately we can prove the following theorem.

\begin{theorem}
\label{thm:cp transformation}
For every constrained critical pair $\ccp{s}{t}{\varphi}$ of $\xR$ and
every substitution $\sigma$ with $\sigma \vDash \varphi$,
(1) $s\sigma = t\sigma$ or (2) there exist a critical pair
$u \approx v$ of $\ov{\xR}$ and a substitution $\delta$ such
that $s\sigma = u\delta$ and $t\sigma = v\delta$.
\end{theorem}

\begin{proof}
Let $\ccp{s}{t}{\varphi}$ be a constrained critical pair of $\xR$
originating from the critical peak
$s = \ell_2 \theta[r_1\theta]_p \L \ell_2\theta[\ell_1\theta]_p
\R r_2\theta = t$ with variants
$\rho_1\colon \crr{\ell_1}{r_1}{\varphi_1}$ and
$\rho_2\colon \crr{\ell_2}{r_2}{\varphi_2}$ of rules in $\xRrc$,
and an mgu $\theta$ of $\ell_2|_p$ and $\ell_1$ where
$p \in \FPos(\ell_2)$.
Moreover $\theta(x) \in \Val \cup \xV$ for all
$x \in \LVar(\rho_1) \cup \LVar(\rho_2)$, and
$\varphi = \varphi_1\theta \land \varphi_2\theta \land
\psi \theta$ with
$\psi = \EC_{\rho_1} \land \EC_{\rho_2}$.
Let $\sigma$ be a substitution with $\sigma \vDash \varphi$. Hence
$\theta\sigma \vDash \varphi_1 \land \varphi_2 \land \psi$
and further $\sigma(\theta(x)) \in \Val$ for all
$x \in \LVar(\rho_1) \cup \LVar(\rho_2)$.
We split $\theta\sigma$ into substitutions $\tau_1$,
$\tau_2$ and $\pi$ as follows:
$\tau_i(x) = x\theta\sigma$ if $x \in \LVar(\rho_i)$
and $\tau_i(x) = x$ otherwise,
for $i \in \SET{1,2}$, and
$\pi(x) = x\theta\sigma$ if $x \in \Dom(\theta\sigma) \setminus
(\LVar(\rho_1) \cup \LVar(\rho_2))$ and
$\pi(x) = x$ otherwise.
From $\theta\sigma \vDash \varphi_1 \land \varphi_2 \land \psi$
and $\Var(\varphi_i) \subseteq \LVar(\rho_i)$
we infer $\tau_i \vDash \varphi_i$ for $i \in \SET{1,2}$.
Since $\Dom(\tau_i) = \LVar(\rho_i)$,
$\ell_i\tau_i \R r_i\tau_i \in \ov{\xR}$ for $i \in \SET{1,2}$.
Furthermore, $\tau_i\pi = \tau_i \cup \pi$ for $i \in \SET{1,2}$.
Hence $\ell_2|_p\tau_2\pi = \ell_2|_p\theta\sigma = \ell_1\theta\sigma =
\ell_1\tau_1\pi$,
implying that $\ell_2|_p\tau_2$ and $\ell_1\tau_1$ are unifiable.
Let $\gamma$ be an mgu of these two terms.
There exists a substitution $\delta$ such that $\gamma\delta = \pi$.
Clearly $p \in \FPos(\ell_2\tau_2)$.
If $p \neq \epsilon$ or $\ell_1\tau_1 \R r_1\tau_1$ and
$\ell_2\tau_2 \R r_2\tau_2$ are not variants, then
$u \approx v$ with
$u = \ell_2\tau_2\gamma[r_1\tau_1\gamma]_p$ and
$v = r_2\tau_2\gamma$ is a critical pair of $\ov{\xR}$. Moreover
$t\sigma = r_2\theta\sigma = r_2\tau_2\pi = r_2\tau_2\gamma\delta =
v\delta$, and similarly $s\sigma = u\delta$. Thus option (2) is satisfied.
If $p = \epsilon$ and $\ell_1\tau_1 \R r_1\tau_1$ and
$\ell_2\tau_2 \R r_2\tau_2$ are variants then
$s\sigma = r_1\tau_1\gamma\delta = r_2\tau_2\gamma\delta = t\sigma$,
fulfilling (1).
\qed
\end{proof}

A TRS (LCTRS) is weakly orthogonal if it is left-linear and
all its (constrained) critical pairs are trivial. Since $\ov{\xR}$ is
left-linear if and only if $\xR$ is left-linear, a
direct consequence of \Cref{thm:cp transformation} is that weak
orthogonality of $\ov{\xR}$ implies weak orthogonality of $\xR$.

Our transformation is not only useful for confluence analysis.

\begin{example}
For the LCTRS $\xR_P$ in the proof of \Cref{thm:undecidable} the
TRS $\ov{\xR}_P$ consists of all unconstrained rules of $\xR_P$
together with
$f(\seq{v}) \to \inter{f(\seq{v})}$
for all $f \in \xFTh \setminus \Val$ and $\seq{v} \in \Val$,
$\m{start} \to \m{test}(\m{alpha}(n),\m{beta}(n))$
for all $n > 0$,
$\m{alpha}(n) \to \alpha_i(\m{alpha}(m))$ and
$\m{beta}(n) \to \beta_i(\m{beta}(m))$
for all $i \in \set[N]$, $n > 0$ and
$m \geqslant 0$ such that $N \cdot m + i = n$.
Termination of the infinite TRS $\ov{\xR}_P$ is easily shown by 
LPO or dependency pairs.
\end{example}

\section{Development Closed Critical Pairs}
\label{sec:dev-closed-ccps}

Using \Cref{thm:CP} we can easily transfer confluence criteria for TRSs to
LCTRSs. Rather than reproving the confluence results reported in
\cite{KN13,WM18,SM23}, in this section we illustrate this by
extending the result of van
Oostrom~\cite{vO97} concerning (almost) development closed critical pairs
from TRSs to LCTRSs. This result subsumes most critical-pair based
confluence criteria, as can be seen in \Cref{fig:venn-diagram} in the
concluding section.

\begin{definition}
\label{def:multi-step rewriting}
Let $\xR$ be an \textup{LCTRS}. The multi-step relation $\MS$ on terms is
defined
inductively as follows: (1) $x \MS x$ for all variables $x$,
(2) $f(\seq{s}) \MS f(\seq{t})$ if $s_i \MS t_i$ with
$1 \leqslant i \leqslant n$,
(3) $\ell\sigma \MS r\tau$ if $\crr{\ell}{r}{\varphi} \in \xRrc$,
$\sigma \vDash \crr{\ell}{r}{\varphi}$ and $\sigma \MS \tau$,
where $\sigma \mto \tau$ denotes $\sigma(x) \MS \tau(x)$ for all variables
$x \in \Dom(\sigma)$.
\end{definition}

\begin{definition}
\label{def:(A)DC}
A critical pair $s \approx t$ is \emph{development closed} if $s \MS t$.
It is \emph{almost development closed} if it is not an overlay
and development closed, or it is an overlay and
$s \MS \cdot {} \La[*] t$.
A \textup{TRS} is called (almost) development closed if all its
critical pairs are (almost) development closed.
\end{definition}

The following result from \cite{vO97} has recently been formalized in
Isabelle~\cite{KM23a,KM23b}.

\begin{theorem}
\label{thm:V}
Left-linear almost development closed \textup{TRS}s are confluent.
\qed
\end{theorem}

We define multi-step rewriting on constrained terms.

\begin{definition}
\label{def:multi-step rewriting constraints}
Let $\xR$ be an \textup{LCTRS}. The multi-step relation $\MS$
on constrained terms is defined inductively as follows:
\begin{enumerate}
\item
$x~\CO{\varphi} \MS x~\CO{\varphi}$ for all variables $x$,
\item
$f(\seq{s})~\CO{\varphi} \MS f(\seq{t})~\CO{\varphi}$ if
$s_i~\CO{\varphi} \MS t_i~\CO{\varphi}$ for
$1 \leqslant i \leqslant n$,
\item
$\ell\sigma~\CO{\varphi} \MS r\tau~\CO{\varphi}$ if
$\rho\colon \crr{\ell}{r}{\psi} \in \xRrc$,
$\sigma(x) \in \Val \cup \Var(\varphi)$ for all $x \in \LVar(\rho)$,
$\varphi$ is satisfiable, $\varphi \Rightarrow \psi\sigma$ is valid, and
$\sigma~\CO{\varphi} \MS \tau~\CO{\varphi}$.
\end{enumerate}
Here $\sigma~\CO{\varphi} \MS \tau~\CO{\varphi}$ denotes
$\sigma(x)~\CO{\varphi} \MS \tau(x)~\CO{\varphi}$ for all variables
$x \in \Dom(\sigma)$. The relation $\sMS$ on constrained terms is defined
as $\sim \cdot \MS \cdot \sim$.
\end{definition}

\begin{example}
Consider the following LCTRS $\xR$ over the theory \textsf{Ints} with the
rules:
\begin{align*}
\m{max}(x,y) &\R x~\CO{x \geqslant y} &
\m{max}(x,y) &\R y~\CO{y \geqslant x}
\end{align*}
Rewriting the term $\m{max}(\m{1} + \m{2}, \m{3} + \m{2})$ to its normal
form $\m{5}$ requires three single steps. These steps can be combined into
a single multi-step $\m{max}(\m{1} + \m{2}, \m{3} + \m{2}) \mto \m{5}$.

The constrained term
$\m{max}(\m{1} + x,\m{3} + y)~\CO{x > \m{3} \land y = \m{1}}$
rewrites in a single multi-step
to its normal form $z~\CO{z = \m{1} + x \land x > \m{3}}$. This
involves the following parts
of~\Cref{def:multi-step rewriting constraints}.
Let $\varphi$ be
$x > \m{3} \land y = \m{1} \land z = \m{1} + x \land z' = \m{3} + y$.
Case (3) gives $\m{1} + x~\CO{\varphi} \MS z~\CO{\varphi}$ and
$\m{3} + y~\CO{\varphi} \MS z'~\CO{\varphi}$.
Using this we obtain
$\m{max}(\m{1} + x,\m{3} + y)~\CO{\varphi} \MS \m{max}(z,z')~\CO{\varphi}$
by case (2). A final application of case (3) yields
$\m{max}(z,z')~\CO{\varphi} \MS z~\CO{\varphi}$. Together with the
equivalences
\begin{align*}
\m{max}(\m{1} + x,\m{3} + y)~\CO{x > \m{3} \land y = \m{1}} &\sim
\m{max}(\m{1} + x,\m{3} + y)~\CO{\varphi} \\
z~\CO{\varphi} &\sim z~\CO{z = 1 + x \land x > 3}
\end{align*} 
we obtain $\m{max}(\m{1} + x,\m{3} + y)~\CO{x > \m{3} \land y = \m{1}}
\sMS z~\CO{z = 1 + x \land x > 3}$.
\end{example}

\Cref{def:(A)DC} is extended to LCTRSs as follows.

\begin{definition}
\label{def:almost development closed}
A constrained critical pair $s \approx t~\CO{\varphi}$ is
\emph{development closed} if
$s \approx t~\CO{\varphi} \sMSab[\geqslant 1]{} u \approx v~\CO{\psi}$
for some trivial $u \approx v~\CO{\psi}$.
A constrained critical pair is \emph{almost development closed} if it is
not an overlay
and development closed, or it is an overlay and
$s \approx t~\CO{\varphi} \sMSab[\geqslant 1]{} \cdot \sRab[\geqslant 2]{*}
u \approx v~\CO{\psi}$ for some trivial $u \approx v~\CO{\psi}$.
An \textup{LCTRS} is called (almost) development closed if all its
constrained critical pairs are (almost) development closed.
\end{definition}

Similar to~\cite{WM18,SM23}, the symbol $\approx$ is treated as a fresh binary
function symbol, resulting in constrained equations whose positions are
addressed in the usual way. Therefore positions below 1 in
$s \approx t~\CO{\varphi}$ refer to subterms of $s$.

\Cref{fig:intuition-transformation} conveys the idea how the
main result (\Cref{thm:almost development closedness})
in this section is obtained. For every critical pair in the
transformed TRS $\ov{\xR}$ there exists a
corresponding constrained critical pair in the original LCTRS $\xR$
(\Cref{thm:CP}). Almost development closure of the constrained critical
pair implies almost development closure of the critical pair
(\Cref{lem:almost development closed}).
Since the rewrite relations of $\xR$ and $\ov{\xR}$
coincide (\Cref{lem:overline relation identity}), we obtain
the confluence of almost development closed left-linear LCTRSs from
the corresponding result in~\cite{vO97}.

\begin{figure}
\centering
\begin{tikzpicture}
  \node (cp) at (-1.6,0) {$s \approx t \in \CP(\ov{\xR})$};
  \node (ccp) at (5,0) {$s' \approx t'~\CO{\varphi} \in \CCP(\xR)$};
  \draw [-implies,double equal sign distance] (0,0) to node [above]
    {\Cref{thm:CP}} (3,0);
  \draw [implies-,double equal sign distance] (0,-1) to node [below]
    {\Cref{lem:almost development closed}} (3,-1);
  \node (bccp) at (3.9,-1) {$\bullet$};
  \draw [->,bend right] (3.6,-0.2) to node [midway] {$\circ$} (3.8,-0.9);
  \draw [->,bend left] (4.2,-0.2) to node [right] {$*$} (4.0,-0.9);
  \node (bcp) at (-2.3,-1) {$\bullet$};
  \draw [->,bend left] (-2.0,-0.2) to node [right] {$*$} (-2.2,-0.9);
  \draw [->,bend right] (-2.6,-0.2) to node [midway] {$\circ$} (-2.4,-0.9);
  \end{tikzpicture}
\caption{Proof idea for \Cref{thm:almost development closedness}.}
\label{fig:intuition-transformation}
\end{figure}

We now present a few technical results that relate rewrite sequences
and multi-steps on (constrained) terms. These prepare for the use of
\Cref{thm:CP} to obtain the confluence of (almost) development closed
LCTRSs. The following lemma is from~\cite{WM18}.

\begin{lemma}
\label{lem:constrained-step-unconstrained}
If $s~\CO{\varphi} \Rb[p] t~\CO{\varphi}$ for a position $p$ and
substitution $\sigma$ with $\sigma \vDash \varphi$ then
$s\sigma \Rb[p] t\sigma$.
\qed
\end{lemma}

\begin{lemma}
\label{lem:key}
\label{lem:repeated-lem-key}
Suppose $s \approx t~\CO{\varphi} \sRab[\geqslanT p]{*}
u \approx v~\CO{\psi}$ with
$\gamma \vDash \varphi$, $\Dom(\gamma) = \Var(\varphi)$ and position
$p$. If $p = 1q$ for a position $q$ then
$s\gamma \Rab[\geqslanT q]{*} u\delta$ and
$t\gamma = v\delta$ for some substitution $\delta$ with
$\delta \vDash \psi$ and $\Dom(\delta) = \Var(\psi)$. If $p = 2q$ for a
position $q$ then $s\gamma = u\delta$ and
$t\gamma \Rab[\geqslanT q]{*} v\delta$ for some substitution
$\delta$ with $\delta \vDash \psi$ and $\Dom(\delta) = \Var(\psi)$.
\end{lemma}

\begin{proof}
We prove the lemma for a single step $s \approx t~\CO{\varphi}
\sRab[\geqslanT p]{} u \approx v~\CO{\psi}$ with
$\gamma \vDash \varphi$, $\Dom(\gamma) = \Var(\varphi)$ and position $p$.
The general statement follows then by an obvious induction argument.
Consider the case $p = 1q$. By definition we can
split the step into
$s \approx t~\CO{\varphi} \sim s' \approx t'~\CO{\varphi'}
\Rab[\geqslanT p]{}
u' \approx t'~\CO{\varphi'} \sim u \approx v~\CO{\psi}$.
Since $\gamma \vDash \varphi$, by definition of $\sim$,
there exists a substitution $\sigma \vDash \varphi'$ with
$\Dom(\sigma) = \Var(\varphi')$ such that
$(s \approx t)\gamma = (s' \approx t')\sigma$.
\Cref{lem:constrained-step-unconstrained} yields
$(s' \approx t')\sigma \Rab[\geqslanT p]{} (u' \approx t')\sigma$.
Again by definition
of $\sim$ we have $(u' \approx t')\sigma = (u \approx v)\tau$ for a
substitution $\tau$ with $\tau \vDash \psi$ and $\Dom(\tau) = \Var(\psi)$.
Hence $s\gamma = s'\sigma \Rab[\geqslanT q]{} u'\sigma = u\tau$
and $t\gamma = t'\sigma = v\tau$.
The case $p = 2q$ is similar.
\qed
\end{proof}

\begin{lemma}
\label{lem:key1}
Suppose $s \approx t~\CO{\varphi} \sRab[\geqslanT p]{*}
u \approx v~\CO{\psi}$ with $\gamma \vDash \varphi$ and position $p$.
If $p = 1q$ for a position $q$ then
$s\gamma \Rab[\geqslanT q]{*} u\delta$ and
$t\gamma = v\delta$ for some substitution
$\delta$ with $\delta \vDash \psi$. If $p = 2q$ for a position $q$ then
$s\gamma = u\delta$ and $t\gamma \sRab[\geqslanT q]{*}
v\delta$ for some substitution $\delta$ with $\delta \vDash \psi$.
\end{lemma}

\begin{proof}
By separating the domain of the substitution to logical and non-logical
variables in \Cref{lem:key} and applying closure under substitution,
we obtain the result.
\qed
\end{proof}

\begin{lemma}
\label{lem:mleadsto to unconstrained}
If $s~\CO{\varphi} \MS t~\CO{\varphi}$ then $s\delta \MS t\delta$ for all
substitutions $\delta \vDash \varphi$.
\end{lemma}

\begin{proof}
We proceed by induction on $\MS$. In case (1) we have
$x~\CO{\varphi} \MS x~\CO{\varphi}$, and $x\delta \MS x \delta$ holds
trivially. In case (2) we have $s = f(\seq{s})$, $t = f(\seq{t})$ and
$s_i~\CO{\varphi} \MS t_i~\CO{\varphi}$ for $1 \leqslant i \leqslant n$.
From the induction hypothesis we obtain $s_i\delta \MS t_i\delta$ for all
$1 \leqslant i \leqslant n$, which further implies $s\delta \MS t\delta$.
In case (3) we have $s = \ell\sigma$ and $t = r\sigma$ for some rule
$\rho\colon \crr{\ell}{r}{\psi} \in \xRrc$,
$\sigma(x) \in \Val \cup \Var(\varphi)$ for all $x \in \LVar(\rho)$,
$\varphi$ is satisfiable, $\varphi \Rightarrow \psi\sigma$ is valid, and
$\sigma(x)~\CO{\varphi} \MS \tau(x)~\CO{\varphi}$ for all
$x \in \Var(\varphi)$. From $\LVar(\rho) \subseteq \Dom(\sigma)$ and
$s = \ell\sigma$ we infer $\Var(\rho) \subseteq \Dom(\sigma)$.
Without loss of generality we assume
$(\Var(s) \cup \Var(t) \cup \Var(\varphi)) \cap \Var(\rho) = \varnothing$
and restrict the domain of $\sigma$ to $\Var(\rho)$ as other variables are
irrelevant. From the induction hypothesis we obtain
$\sigma(x)\delta \MS \tau(x)\delta$ for all $x \in \Var(\varphi)$.
Moreover, since $\delta \vDash \varphi$ we have
$\delta \vDash \psi\sigma$ and thus also $\sigma\delta \vDash \psi$.
Therefore $s\delta = \ell\sigma\delta \MS r\tau\delta = t\delta$ as
desired.
\qed
\end{proof}

\begin{lemma}
\label{lem:keymto}
If $s\approx t~\CO{\varphi} \sMSab[\geqslant 1]{} u \approx v~\CO{\psi}$
then for all substitutions $\sigma \vDash \varphi$ with
$\Dom(\sigma) = \Var(\varphi)$ there exists $\delta \vDash \psi$
with $\Dom(\delta) = \Var(\psi)$ such that $s\sigma \MS u\delta$ and
$t\sigma = v\delta$.
\end{lemma}

\begin{proof}
By unfolding the definition of $\sMS$ we obtain
$s \approx t~\CO{\varphi} \sim s' \approx t'~\CO{\varphi'}
\MS_{\geqslant 1} u' \approx v'~\CO{\varphi'} \sim u \approx v~\CO{\psi}$.
Let $\sigma$ be a substitution with $\sigma \vDash \varphi$. From the
definition of $\sim$ we obtain a substitution $\tau$ such that
$\tau \vDash \varphi'$, $\Dom(\tau) = \Var(\varphi')$, $s\sigma = s'\tau$
and $t\sigma = t'\tau$. As all contracted redexes in the multi-step
$s' \approx t'~\CO{\varphi'}$ are below the position $1$, this
corresponds to case (2) in~\Cref{def:multi-step rewriting constraints}
with $s'$ and $t'$ being the first and second argument of $\approx$. Hence
$s'~\CO{\varphi'} \MS u'~\CO{\varphi'}$ and $t' = v'$. We therefore obtain
$t'\tau = v'\tau$ and $s'\tau \MS u'\tau$ from
\Cref{lem:mleadsto to unconstrained}. From the equivalence
$u' \approx v'~\CO{\varphi'} \sim u \approx v~\CO{\psi}$ together with
$\tau \vDash \varphi'$ and $\Dom(\tau) = \varphi'$ we obtain a
substitution $\gamma$ such that $\gamma \vDash \psi$,
$\Dom(\gamma) = \Var(\psi)$, $u'\tau = u\gamma$ and $v'\tau = v\gamma$.
Hence
$s\sigma = s'\tau \MS u'\tau = u\gamma$ and
$t\sigma = t'\tau = v'\tau = v\gamma$.
\qed
\end{proof}

\begin{lemma}
\label{lem:key2}
If $s\approx t~\CO{\varphi} \sMSab[\geqslant 1]{} u \approx v~\CO{\psi}$
then for all substitutions $\sigma \vDash \varphi$ there exists
$\delta \vDash \psi$ such that $s\sigma \MS u\delta$ and
$t\sigma = v\delta$.
\end{lemma}

\begin{proof}
This follows by dropping the restriction on the substitution in
\Cref{lem:keymto}, exactly like in the proof of \Cref{lem:key1}.
\qed
\end{proof}

\begin{lemma}
\label{lem:almost development closed}
If a constrained critical pair $s \approx t~\CO{\varphi}$ is almost
development closed then for all substitutions $\sigma$ with
$\sigma \vDash \varphi$ we have
$s\sigma \MS \cdot \mr{\prescript{*}{}{\L}} t\sigma$.
\end{lemma}

\begin{proof}
Let $s \approx t~\CO{\varphi}$ be an almost development closed constrained
critical pair, and $\sigma \vDash \varphi$ some substitution. From
\Cref{def:almost development closed} we obtain
\begin{align}
\label{adc}
s \approx t~\CO{\varphi} \sMS_{\geqslant 1} u' \approx v'~\CO{\psi'}
\sRab[\geqslant 2]{*} u \approx v~\CO{\psi}
\end{align}
where $u\tau = v\tau$ for all $\tau \vDash \psi$ for some constrained
term $u' \approx v'~\CO{\psi'}$. We apply \Cref{lem:key2} to the
first step in \eqref{adc}.
This yields a substitution $\delta$ where
$s\sigma \MS u' \delta$, $t\sigma = v'\delta$
and $\delta \vDash \psi'$.
For the second part of \eqref{adc} we use
\Cref{lem:key1} and obtain $v'\delta \Ra[*] v \gamma$,
$u'\delta = u\gamma$ for some $\gamma \vDash \psi$.
Moreover we have $u\gamma = v\gamma$. Hence
$s\sigma \MS u'\delta = u\gamma = v\gamma
\mr{\prescript{*}{}{\L}} v'\delta = t\sigma$.
\qed
\end{proof}

\begin{theorem}
\label{thm:almost development closedness}
If an \textup{LCTRS} $\xR$ is almost development closed then
so is $\ov{\xR}$.
\end{theorem}

\begin{proof}
Take any critical pair $s \approx t$ from $\ov{\xR}$.
From \Cref{thm:CP} we know that there exists a
constrained critical pair $s' \approx t'~\CO{\varphi}$
in $\xR$ where $s'\sigma = s$ and $t'\sigma = t$
for some $\sigma \vDash \varphi$.
Since the constrained critical pair must be almost development closed,
\Cref{lem:almost development closed} yields
$s = s'\sigma \MS \cdot \La[*] t'\sigma = t$
if it is an overlay and $s = s'\sigma \MS t'\sigma = t$ otherwise.
This proves that $\ov{\xR}$ is almost development closed.
\qed
\end{proof}

Interestingly, the converse does not hold, as seen in the following
example.

\begin{example}
\label{exa:trs-adc-notin-lctrs}
Consider the LCTRS $\xR$ over the theory \textsf{Ints} with
the rules
\begin{align*}
\m{f}(x) &\R \m{g}(x) &
\m{g}(x) &\R \m{h}(\m{2})~\CO{x = \m{2}z}\\
\m{f}(x) &\R \m{h}(x)~\CO{\m{1} \leqslant x \leqslant \m{2}} &
\m{g}(x) &\R \m{h}(\m{1})~\CO{x = \m{2}z + \m{1}}
\end{align*}
The TRS $\ov{\xR}$ consists of the rules
\begin{align*}
\m{f}(x) &\R \m{g}(x) &
\m{f}(\m{1}) &\R \m{h}(\m{1}) &
\m{g}(n) &\R \m{h}(\m{1}) \quad \text{for all odd $n \in \ZZ$} \\
&&
\m{f}(\m{2}) &\R \m{h}(\m{2}) &
\m{g}(n) &\R \m{h}(\m{2}) \quad \text{for all even $n \in \ZZ$}
\end{align*}
and has two (modulo symmetry) critical pairs
$\m{g}(\m{1}) \approx \m{h}(\m{1})$ and
$\m{g}(\m{2}) \approx \m{h}(\m{2})$.
Since $\m{g}(\m{1}) \MS \m{h}(\m{1})$ and $\m{g}(\m{2}) \MS \m{h}(\m{2})$,
$\ov{\xR}$ is almost development closed. The constrained critical pair
$\m{g}(x) \approx \m{h}(x)~\CO{\m{1} \leqslant x \leqslant \m{2}}$
is not almost development closed, since it is a normal form with
respect to the rewrite relation on constrained terms.
\end{example}

This also makes intuitive sense, since a rewrite step
$s \approx t~\CO{\varphi} \Rs u \approx v~\CO{\psi}$ implies that the same
step can be taken on all instances $s\sigma \approx t\sigma$ where
$\sigma \vDash \varphi$. However it may be the case, like in the above
example, that different instances of the constrained critical pair
require different steps to obtain a closing sequence, which cannot
directly be modeled using rewriting on constrained terms.

Since left-linearity of $\ov{\xR}$ is preserved,
the following corollary is obtained from Theorems~\ref{thm:V}
and~\ref{thm:almost development closedness}.
In fact $\xR$ only has to be linear in the variables $x \notin \LVar$,
since that is sufficient for $\ov{\xR}$ to be linear.

\begin{corollary}
\label{cor:almost-dev-closed}
Left-linear almost development closed \textup{LCTRS}s are confluent. \qed
\end{corollary}

\begin{example}
\label{exa:almost-dev-closed}
The LCTRS $\xR$ over the theory $\m{Ints}$ with the rules
\begin{align*}
\m{f}(x,y) &\R
\m{h}(\m{g}(y,\m{2} \cdot \m{2}))~\CO{x \leqslant y \land y = \m{2}} &
\m{g}(x,y) &\R \m{g}(y,x) & \m{h}(x) &\R x \\
\m{f}(x,y) &\R \m{c}(\m{4},x)~\CO{y \leqslant x} &
\m{c}(x,y) &\R \m{g}(\m{4},\m{2})~\CO{x \neq y}
\end{align*}
admits the two constrained critical pairs (with simplified
constraints)
\begin{align*}
\m{h}(\m{g}(y,\m{2} \cdot \m{2})) &\approx \m{c}(\m{4},x)~\CO{\varphi} &
\m{c}(\m{4},x) &\approx \m{h}(\m{g}(y,\m{2} \cdot \m{2}))~\CO{\varphi}
\end{align*}
Both are almost development closed:
\begin{align*}
&\m{h}(\m{g}(y,\m{2} \cdot \m{2})) \approx \m{c}(\m{4},x)~\CO{\varphi} &
&\m{c}(\m{4},x) \approx \m{h}(\m{g}(y,\m{2} \cdot \m{2}))~\CO{\varphi} \\
\sMSab[\geqslant 1]{}~ &\m{g}(\m{4},\m{2}) \approx
\m{c}(\m{4},x)~\CO{x = \m{2}} &
\sMSab[\geqslant 1]{}~ &\m{g}(\m{4},\m{2}) \approx
\m{h}(\m{g}(y,\m{2} \cdot \m{2}))~\CO{y = \m{2}} \\
\sRab[\geqslant 2]{}~ &\m{g}(\m{4},\m{2}) \approx
\m{g}(\m{4},\m{2})~\CO{\m{true}}&
\sRab[\geqslant 2]{*}~ &\m{g}(\m{4},\m{2}) \approx
\m{g}(\m{4},\m{2})~\CO{\m{true}}
\end{align*}
Here $\varphi$ is the constraint $x = y \land y = \m{2}$.
Hence $\xR$ is almost development closed. Since $\xR$ is 
left-linear, confluence follows by~\Cref{cor:almost-dev-closed}.
\end{example}

\section{Parallel Critical Pairs}
\label{sec:pcps}

In this section we extend the confluence result by Toyama~\cite{T81}
based on parallel critical pairs to LCTRSs. Recently there is
a renewed interest in this result;
Shintani and Hirokawa proved in~\cite{SH22} that it subsumes Toyama's
later confluence result in~\cite{T88}. The latter was already lifted
to LCTRSs in~\cite{SM23} and is also subsumed by
\Cref{cor:almost-dev-closed}. The result of Toyama~\cite{T81} is
a proper extension of the confluence criterion on parallel critical pairs
by Gramlich~\cite{G96}. In the sequel we mainly follow the
notions from~\cite{SH22}.

\begin{definition}
\label{def:parallel rewriting}
Let $\xR$ be an \textup{LCTRS}. The parallel rewrite relation $\pto$ on
terms is defined inductively as follows:
\begin{enumerate}
\item
$x \pto x$ for all variables $x$,
\item
$f(\seq{s}) \pto f(\seq{t})$ if $s_i \pto t_i$ for
$1 \leqslant i \leqslant n$,
\item
$\ell\sigma \pto r\sigma$ if $\crr{\ell}{r}{\varphi} \in \xRrc$ and
$\sigma \vDash \crr{\ell}{r}{\varphi}$
\end{enumerate}
We extend $\pto$ to constrained terms inductively as follows:
\begin{enumerate}
\item
$x~\CO{\varphi} \pto x~\CO{\varphi}$ for all variables $x$,
\item
$f(\seq{s})~\CO{\varphi} \pto f(\seq{t})~\CO{\varphi}$ if
$s_i~\CO{\varphi} \pto t_i~\CO{\varphi}$ for $1 \leqslant i \leqslant n$,
\item
$\ell\sigma~\CO{\varphi} \pto r\sigma~\CO{\varphi}$ if
$\rho\colon \crr{\ell}{r}{\psi} \in \xRrc$,
$\sigma(x) \in \Val \cup \Var(\varphi)$ for all $x \in \LVar(\rho)$,
$\varphi$ is satisfiable and $\varphi \Rightarrow \psi\sigma$ is valid.
\end{enumerate}
The parallel rewrite relation $\spto$ on constrained terms is defined as
$\sim \cdot \pto \cdot \sim$.
\end{definition}

Let $s$ be a term and $P \subseteq \Pos(s)$ be a set of parallel
positions. Given terms $t_p$ for $p \in P$, we denote by
$s[t_p]_{p \iN P}$ the simultaneous replacement of the terms at
position $p \in P$ in $s$ by $t_p$. We recall the definition of parallel
critical pairs for TRSs.

\begin{definition}
\label{def:pcp}
Let $\xR$ be a \textup{TRS}, $\rho\colon \ell \R r$ a rule in $\xR$, and
$P \subseteq \FPos(\ell)$ a non-empty set of parallel positions.
For every $p \in P$ let $\rho_p\colon \ell_p \R r_p$ be a variant of a
rule in $\xR$. The peak $\ell\sigma[r_p\sigma]_{p \iN P} \rpto
\ell\sigma \Rb[\epsilon,\xR] r\sigma$ forms a
\emph{parallel critical pair}
$\ell\sigma[r_p\sigma]_{p \iN P} \approx r\sigma$ if the following
conditions are satisfied:
\begin{enumerate}
\item
$\Var(\rho_1) \cap \Var(\rho_2) = \varnothing$ for different rules
$\rho_1$ and $\rho_2$ in $\SET{\rho} \cup \SET{\rho_p \mid p \in P}$,
\item
$\sigma$ is an mgu of $\SET{\ell_p \approx \ell|_p \mid p \in P}$,
\item
if $P = \SET{\epsilon}$ then $\rho_\epsilon$ is not a variant of $\rho$.
\end{enumerate}
The set of all constrained parallel critical pairs of $\xR$ is denoted by
$\PCP(\xR)$.
\end{definition}

We lift this notion to the constrained setting and define it for LCTRSs.

\begin{definition}
\label{def:cpcp}
Let $\xR$ be an \textup{LCTRS}, $\rho\colon \ell \R r~\CO{\varphi}$ a
rule in $\xRrc$, and $P \subseteq \FPos(\ell)$ a non-empty set of parallel
positions. For every $p \in P$ let
$\rho_p\colon \ell_p \R r_p~\CO{\varphi_p}$ be a variant of a rule in
$\xRrc$. Let $\psi = \EC_{\rho} \land
\bigwedge_{p \iN P} \EC_{\rho_p}$ and $\Phi = \varphi\sigma \land
\psi\sigma \land \bigwedge_{p \iN P} \varphi_p\sigma$. The peak
$\ell\sigma[r_p\sigma]_{p \iN P}~\CO{\Phi} \rpto \ell\sigma~\CO{\Phi}
\Rb[\epsilon,\xR] r\sigma~\CO{\Phi}$ forms a
\emph{constrained parallel critical pair}
$\ell\sigma[r_p\sigma]_{p \iN P} \approx r\sigma~\CO{\Phi}$ if the
following conditions are satisfied:
\begin{enumerate}
\item
$\Var(\rho_1) \cap \Var(\rho_2) = \varnothing$ for different rules
$\rho_1$ and $\rho_2$ in $\SET{\rho} \cup \SET{\rho_p \mid p \in P}$,
\item
$\sigma$ is an mgu of $\SET{\ell_p = \ell|_p \mid p \in P}$
such that $\sigma(x) \in \Val \cup \xV$ for all
$x \in \LVar(\rho) \cup \bigcup_{p \iN P} \LVar(\rho_p)$,
\item
$\varphi\sigma \land \bigwedge_{p \iN P} \varphi_p\sigma$ is satisfiable,
\item
if $P = \SET{\epsilon}$ then $\rho_\epsilon$ is not a variant of $\rho$
or $\Var(r) \nsubseteq \Var(\ell)$.
\end{enumerate}
A constrained peak forming a constrained parallel critical pair is
called a \emph{constrained parallel critical peak}. The set of all
constrained parallel critical pairs of $\xR$ is denoted by $\CPCP(\xR)$.
\end{definition}

For a term $t$ and a set of parallel positions $P$ in $t$, we write
$\Var(t,P)$ to denote $\bigcup_{p \iN P} \Var(t|_p)$.
For a set of parallel positions $P$ we denote by $\pto^P$ that each
rewrite step obtained in case~(3) of
\Cref{def:parallel rewriting}
is performed at a position $p \in P$ and no two steps share a position.
Moreover, for a set of parallel positions $P$ and a position $q$ we
denote by $\pto_{\geqslanT q}^P$ that $p \geqslant q$ for all $p \in P$.

\begin{definition}
\label{def:toy81-1}
A critical pair $s \approx t$ is \emph{1-parallel closed}
if $s \pto \cdot \La[*] t$. A \textup{TRS} is 1-parallel closed if all its
critical pairs are 1-parallel closed. A parallel critical pair
$\ell\sigma[r_p\sigma]_{p \iN P} \approx r\sigma$ originating from
the peak
$\ell\sigma[r_p\sigma]_{p \iN P} \rpto \ell\sigma \Rb[\epsilon] r\sigma$
is \emph{2-parallel closed} if
there exists a term $v$ and a set of parallel positions $Q$ such that
$\ell\sigma[r_p\sigma]_{p \iN P} \Ra[*] v \rptoa[Q] r\sigma$ with
$\Var(v,Q) \subseteq \Var(\ell\sigma,P)$.
A \textup{TRS} is 2-parallel closed if all its parallel critical pairs are
2-parallel closed.
A \textup{TRS} is parallel closed if it is 1-parallel closed and
2-parallel closed.
\end{definition}

The following result from \cite{T81} has recently
been formalized in Isabelle~\cite{HKST24}.

\begin{theorem}
\label{thm:toy81}
Left-linear parallel closed \textup{TRS}s are confluent.
\qed
\end{theorem}

In the remainder of this section we extend this result
to LCTRSs.
To this end we introduce the notion
$\TVar(t,\varphi) = \Var(t) \setminus \Var(\varphi)$ denoting the set of
non-logical variables in term $t$ with respect to the logical
constraint $\varphi$.
We restrict this to non-logical variables in subterms
below a set of parallel positions $P$ in $t$:
$\TVar(t,\varphi,P) = \bigcup_{p \iN P} \TVar(t|_p,\varphi)$.

\begin{definition}
\label{def:toy81LCTRS-1}
A constrained critical pair $s \approx t~\CO{\varphi}$ is
\emph{1-parallel closed} if $s \approx t~\CO{\varphi}
\sptoab[\geqslant 1]{} \cdot \sRab[\geqslant 2]{*} u \approx v~\CO{\psi}$
for some trivial $u \approx v~\CO{\psi}$. An \textup{LCTRS} is 1-parallel
closed if all its constrained critical pairs are 1-parallel closed. A
constrained parallel critical pair
$\ell\sigma[r_p\sigma]_{p \iN P} \approx r\sigma~\CO{\varphi}$
is \emph{2-parallel closed} if there exists a set of parallel positions
$Q$ such that
\begin{align*}
\ell\sigma[r_p\sigma]_{p \iN P} \approx r\sigma~\CO{\varphi}
\sptoab[\geqslant 2]{Q} \cdot \sRab[\geqslant 1]{*} u \approx v~\CO{\psi}
\end{align*}
for some trivial $u \approx v~\CO{\psi}$ and
$\TVar(v,\psi,Q) \subseteq \TVar(\ell\sigma,\varphi,P)$.
An \textup{LCTRS} is 2-parallel closed if all its constrained parallel
critical pairs are 2-parallel closed. An \textup{LCTRS} is parallel closed
if it is 1-parallel closed and 2-parallel closed.
\end{definition}

Recall from \Cref{sec:prelims} that our definition of $\sim$ differs from
the equivalence relation $\sim'$ defined in~\cite{KN13,SM23}.
The change is necessary for the variable condition of 2-parallel
closedness to make sense, as illustrated in the following example.

\begin{example}
\label{exa:var-cond-equivalence}
Consider the (LC)TRS consisting of the rules
\begin{align*}
\m{f}(\m{g}(x),y) &\to \m{f}(\m{b},y) &
\m{g}(x) &\to \m{a} &
\m{f}(\m{a},x) &\to x & \m{f}(\m{b},x) &\to x
\end{align*}
The peak
$\m{f}(\m{a},y)~\CO{\m{true}} \rptoa[\SET{1}]
\m{f}(\m{g}(x),y)~\CO{\m{true}} \to \m{f}(\m{b},y)~\CO{\m{true}}$
gives rise to the (constrained) parallel critical pair
$\m{f}(\m{a},y) \approx \m{f}(\m{b},y)~\CO{\m{true}}$.
Using $\sim'$ we have
\begin{align*}
\m{f}(\m{a},y) \approx \m{f}(\m{b},y)~\CO{\m{true}}
\pto_{\geqslant 2}^{\SET{\epsilon}} \cdot \Rab[\geqslant 1]{*}
y \approx y~\CO{\m{true}} \sim' x \approx x~\CO{\m{true}}
\end{align*}
and the variable condition $\TVar(x,\m{true},\SET{\epsilon}) \subseteq
\TVar(\m{f}(\m{g}(x),y),\m{true},\SET{1})$
holds.
Since the system has no logical constraints it can also be analyzed in
the TRS setting.
Following \Cref{def:toy81-1} we would have to check the variable condition
$\Var(y,\SET{\epsilon}) \subseteq \Var(\m{f}(\m{g}(x),y),\SET{1})$,
which does not hold.
Using $\sim$ resolves this difference, since
$y \approx y~\CO{\m{true}} \not\sim x \approx x~\CO{\m{true}}$.
So the conditions in \Cref{def:toy81LCTRS-1} reduce to the ones in
\Cref{def:toy81-1} for TRSs.
\end{example}

In \Cref{thm:CP} in \Cref{sec:transformation} we related critical
pairs of the transformed TRS to constrained critical pairs of the
originating LCTRS. The following theorem does the same for parallel
critical pairs. 

\begin{theorem}
\label{thm:PCP}
For every parallel critical pair $s \approx t$ of $\ov{\xR}$
there exists a constrained parallel critical pair
$\ccp{s'}{t'}{\varphi'}$ of $\xR$ and a substitution $\gamma$ such that
$s = s'\gamma$, $t = t'\gamma$ and $\gamma \vDash \varphi'$.
\end{theorem}

\begin{proof}
Let $s \approx t$ be a parallel critical pair of $\ov{\xR}$, originating
from the parallel critical peak $\ell\mu\sigma[r_p\nu_p\sigma]_{p \iN P}
\rpto \ell\mu\sigma = \ell\mu\sigma[\ell_p\nu_p\sigma]_{p \iN P}
\Rb[\epsilon]
r\mu\sigma$ with variants $\rho\colon \crr{\ell}{r}{\varphi}$ and
$\rho_p\colon \crr{\ell_p}{r_p}{\varphi_p}$ for $p \in P$ of rules in
$\xRrc$ without shared variables,
$\psi = \EC_\rho$ and
$\psi_p = \EC_{\rho_p}$ for $p \in P$.
Furthermore, $\Dom(\nu_p) = \LVar(\rho_p)$ for
$p \in P$, $\Dom(\mu) = \LVar(\rho)$,
$\nu_p \vDash \varphi_p \land \psi_p$ for $p \in P$,
$\mu \vDash \varphi \land \psi$, $p \in \FPos(\ell\mu)$, and $\sigma$ is
an mgu of
$\SET{\ell\mu|_p \approx \ell_p\nu_p \mid p \in P}$. Moreover, if
$P = \SET{\epsilon}$ then $\ell_\epsilon\nu_\epsilon \R
r_\epsilon\nu_\epsilon~\CO{\varphi_\epsilon\nu_\epsilon}$ and
$\ell\mu \R r\mu~\CO{\varphi\mu}$ are not variants. Define the
substitution $\tau$ as $\bigcup\,\SET{\nu_p \mid p \in P} \uplus \mu$.
Clearly, $\ell_p\tau = \ell_p\nu_p$ and $r_p\tau = r_p\nu_p$ for
$p \in P$, $\ell\tau = \ell\mu$, $r\tau = r\mu$,
$\tau \vDash \varphi \land \psi$ and
$\tau \vDash \varphi_p \land \psi_p$ for all $p \in P$.
Hence the given peak can be written as
$\ell\tau\sigma[r_p\tau\sigma]_{p \iN P} \rpto \ell\tau\sigma =
\ell\tau\sigma[\ell_p\tau\sigma]_{p \iN P} \Rb[\epsilon] r\tau\sigma$ with
$\tau \vDash \varphi''$ where
\[
\varphi'' = \varphi \land \EC_\rho \land
\bigwedge_{p \iN P} (\varphi_p \land \EC_{\rho_p})
\]
Since $\ell|_p\tau\sigma = \ell_p\tau\sigma$ for all $p \in P$ there
exists an mgu $\delta$ of $\SET{\ell|_p = \ell_p \mid p \in P}$ and a
substitution $\gamma$ such that $\delta\gamma = \tau\sigma$. Let
$s' = \ell\delta[r_p\delta]_{p \iN P}$ and $t' = r\delta$. We claim that
this results in the constrained parallel critical pair
$\ccp{s'}{t'}{\varphi''\delta}$. Condition (1) of \Cref{def:cpcp} is
trivially satisfied. We obtain $P \subseteq \FPos(\ell)$
because $P \subseteq \FPos(\ell\mu)$,
$\mu(x) \in \Val$ for every $x \in \Dom(\mu)$, and
$\m{root}(\ell\mu|_p) = \m{root}(\ell_p\nu) \in \xF \setminus \Val$ for
all $p \in P$. For condition (2) it
remains to show that $\delta(x) \in \Val \cup \xV$ for all
$x \in \LVar(\rho) \cup \bigcup_{p \iN P} \LVar(\rho_p)$. Suppose to the
contrary that $\m{root}(\delta(x)) \in \xF \setminus \Val$
for some $x \in \LVar(\rho) \union \bigcup_{p \iN P} \LVar(\rho_p)$. Then
$\m{root}(\delta(x)) = \m{root}(\gamma(\delta(x))) =
\m{root}(\sigma(\tau(x))) \in \xF \setminus \Val$, which contradicts
$\tau \vDash \varphi''$. Condition (3) follows from the identity
$\delta\gamma = \tau\sigma$ together with $\tau \vDash \varphi''$ which
imply $\delta\gamma \vDash \varphi''$ and thus $\varphi''\delta$ is
satisfiable. Hence also
$\varphi\delta \land \bigwedge_{p \iN P} \varphi_p\delta$ is satisfiable.
It remains to show condition (4), so let $P = \SET{\epsilon}$ and further
assume that $\rho_\epsilon$ and $\rho$ are variants. So there exists a
variable renaming $\pi$ such that $\rho_\epsilon\pi = \rho$. In
particular, $\ell_\epsilon\pi = \ell$ and $r_\epsilon\pi = r$. We show
$\tau(\pi(x)) = \pi(\tau(x))$ for all $x \in \Var(\ell_\epsilon)$. Let
$x \in \Var(\ell_\epsilon)$. If $x \in \LVar(\rho_\epsilon) = \Dom(\nu)$
then $\tau(x) = \nu(x) \in \Val$. Moreover,
$\pi(x) \in \LVar(\rho) = \Dom(\mu)$ and thus
$\tau(\pi(x)) = \mu(\pi(x)) \in \Val$. Since $\ell_\epsilon\tau$ and
$\ell\tau$ are unifiable, $\pi(\tau(x)) = \tau(x) = \tau(\pi(x))$.
If $x \notin \LVar(\rho_\epsilon)$ then $\tau(x) = x$,
$\pi(x) \notin \LVar(\rho)$ and similarly
$\tau(\pi(x)) = \pi(x) = \pi(\tau(x))$.
All in all, $\ell_\epsilon\tau\pi = \ell_\epsilon\pi\tau = \ell\tau$.
Now, if $\Var(r_\epsilon) \subseteq \Var(\ell_\epsilon)$ then we obtain
$r_\epsilon\tau\pi = r_\epsilon\pi\tau = r\tau$, contradicting the fact
that $\ell_\epsilon\nu \R r_\epsilon\nu$ and $\ell\mu \R r\mu$ are not
variants.  We conclude that $\ccp{s'}{t'}{\varphi''\delta}$ is a
constrained parallel critical pair of $\xR$. So we can take
$\varphi' = \varphi''\delta$. Clearly,
$s = s'\gamma$ and $t = t'\gamma$. Moreover, $\gamma \vDash \varphi'$
since $\varphi'\gamma = \varphi''\tau\sigma = \varphi''\tau$ and
$\tau \vDash \varphi''$.
\qed
\end{proof}

\begin{lemma}
\label{lem:pto to unconstrained}
If $s~\CO{\varphi} \pto^P t~\CO{\varphi}$ then $s\delta \pto^P t\delta$
for all substitutions $\delta$ with $\delta \vDash \varphi$.
\end{lemma}

\begin{proof}
We proceed by induction on $\pto$. In case (1) we have
$x~\CO{\varphi} \pto^\varnothing x~\CO{\varphi}$, and
$x\delta \pto^\varnothing x\delta$ holds trivially. In case (2) we
have $s = f(\seq{s})$, $t = f(\seq{t})$,
$s_i~\CO{\varphi} \pto^{P_i} t_i~\CO{\varphi}$ for
$1 \leqslant i \leqslant n$
and $P = \SET{ip \mid \text{$1 \leqslant i \leqslant n$ and $p \in P_i$}}$.
From the induction hypothesis we obtain $s_i\delta \pto^{P_i} t_i\delta$
for all $1 \leqslant i \leqslant n$ and thus also
$s\delta \pto^P t\delta$. In case (3) we have
$s \pto^{\SET{\epsilon}} t$
with $s = \ell\sigma$ and $t = r\sigma$ for some rule
$\rho\colon \crr{\ell}{r}{\psi} \in \xRrc$ and substitution $\sigma$
such that $\sigma(x) \in \Val \cup \Var(\varphi)$ for all
$x \in \LVar(\rho)$, $\varphi$ is satisfiable and
$\varphi \Rightarrow \psi\sigma$ is valid.
From $\LVar(\rho) \subseteq \Dom(\sigma)$ and $s = \ell\sigma$
we infer
$\Var(\rho) \subseteq \Dom(\sigma)$. Without loss of generality we assume
$(\Var(s) \cup \Var(t) \cup \Var(\varphi)) \cap \Var(\rho) = \varnothing$
and restrict the domain of $\sigma$ to $\Var(\rho)$. Moreover, from
$\delta \vDash \varphi$ we infer
$\delta \vDash \psi\sigma$ and thus also $\sigma\delta \vDash \psi$.
Therefore
$s\delta = \ell\sigma\delta \pto^{\SET{\epsilon}} r\sigma\delta = t\delta$
as desired.
\qed
\end{proof}

The following lemma is proved by replaying
the proof of~\Cref{lem:keymto}, using
\Cref{def:parallel rewriting} and \Cref{lem:pto to unconstrained} with the
respective adaptions to the parallel rewrite relation.

\begin{lemma}
\label{lem:keypto}
If $s\approx t~\CO{\varphi} \sptoab[\geqslant 1]{P} u \approx v~\CO{\psi}$
then for all substitutions $\sigma \vDash \varphi$ with
$\Dom(\sigma) = \Var(\varphi)$ there exists a $\delta \vDash \psi$
with $\Dom(\delta) = \Var(\psi)$ such that $s\sigma \pto^P u\delta$ and
$t\sigma = v\delta$.
\qed
\end{lemma}

\begin{lemma}
\label{lem:key3}
If $s \approx t~\CO{\varphi} \sptoab[\geqslant 1]{P} u \approx v~\CO{\psi}$
then for all substitutions $\sigma \vDash \varphi$ there exists a
substitution $\delta$ such that
$\delta \vDash \psi$, $s\sigma \pto^P u\delta$ and
$t\sigma = v\delta$.
\end{lemma}

\begin{proof}
This follows from \Cref{lem:keypto}, similar to the proof of
\Cref{lem:key1}. \qed
\end{proof}

\begin{lemma}
\label{lem:1pc-unconstrained}
If a constrained critical pair $s \approx t~\CO{\varphi}$ is
1-parallel closed then
$s\sigma \pto \cdot \mr{\prescript{*}{}{\L}} t\sigma$ for all
substitutions $\sigma$ with $\sigma \vDash \varphi$.
\end{lemma}

\begin{proof}
Let $s \approx t~\CO{\varphi}$ be a 1-parallel closed constrained critical
pair. So
\begin{align}
\label{j1}
s \approx t~\CO{\varphi} \sptoab[\geqslant 1]{}
s' \approx t'~\CO{\varphi'} \sRab[\geqslant 2]{*} u \approx v~\CO{\psi}
\end{align}
with trivial $u \approx v~\CO{\psi}$. Let $\sigma \vDash \varphi$. To the
first part of \eqref{j1} we apply \Cref{lem:keypto}. This yields a
substitution $\delta \vDash \varphi'$ such that $s\sigma \pto^P s'\delta$
and $t\sigma = t'\delta$. According to \Cref{lem:key1}
the second part of \eqref{j1} yields $s'\delta = u\gamma$ and
$t'\delta \Ra[*] v\gamma$ for some substitution $\gamma$ with
$\gamma \vDash \psi$. We obtain $u\gamma = v\gamma$ from the triviality of
$u \approx v~\CO{\psi}$. Combining the above facts yields the desired
$s\sigma \pto^P u\gamma = v\gamma \La[*] t\sigma$.
\qed
\end{proof}

\begin{lemma}
\label{lem:2pc-unconstrained}
If a constrained parallel critical pair
$s = \ell\sigma'[r_p\sigma']_{p \iN P} \approx r\sigma' = t~\CO{\varphi}$
is 2-parallel closed then there exist a term $v$ and a set $Q$ of
parallel positions such that $s\sigma \Ra[*] v \rptoa[Q] t\sigma$
and $\Var(v,Q) \subseteq \Var(\ell\sigma'\sigma,P)$
for all substitutions $\sigma$ with $\sigma \vDash \varphi$.
\end{lemma}

\begin{proof}
Let
$s = \ell\sigma'[r_p\sigma']_{p \iN P} \approx r\sigma' = t~\CO{\varphi}$
be a 2-parallel closed constrained critical pair. So
\begin{align}
\label{j2}
s \approx t~\CO{\varphi} \sptoab[\geqslant 2]{Q}
s' \approx t'~\CO{\varphi'} \sRab[\geqslant 1]{*}
u' \approx v'~\CO{\psi}
\end{align}
with trivial $u' \approx v'~\CO{\psi}$ and
$\TVar(v',\psi,Q) \subseteq \TVar(\ell\sigma',\varphi,P)$. Consider a
substitution $\sigma$ and let $\sigma \vDash \varphi$. We split $\sigma$
into $\sigma = \sigma_{\mth} \cup \sigma_{\mte} =
\sigma_{\mth}\sigma_{\mte}$ such that
$\Dom(\sigma_{\mth}) \cap \Dom(\sigma_{\mte}) = \varnothing$ and
$\Dom(\sigma_{\mth}) = \Var(\varphi)$.
From
$\sigma \vDash \varphi$ we obtain $\sigma_{\mth} \vDash \varphi$.
To the first part of \eqref{j2} we apply \Cref{lem:keypto}.
This yields a substitution $\delta \vDash \varphi'$ with
$\Dom(\delta) = \Var(\varphi')$ such that $s\sigma_{\mth} = s'\delta$ and
$t\sigma_{\mth} \pto^Q t'\delta$.
Applying \Cref{lem:repeated-lem-key} to the second part of \eqref{j2}
yields
$s'\delta \Ra[*] u'\gamma$ and $t'\delta = v'\gamma$ for some substitution
$\gamma$ with $\gamma \vDash \psi$ and $\Dom(\gamma) = \Var(\psi)$. We
obtain $u'\gamma = v'\gamma$ from the triviality of
$u' \approx v'~\CO{\psi}$. Combining the above facts yields
$s\sigma = s\sigma_{\mth}\sigma_{\mte} \Ra[*]
u'\gamma\sigma_{\mte} = v'\gamma\sigma_{\mte} \rptoa[Q]
t\sigma_{\mth}\sigma_{\mte} = t\sigma$ by closure
of rewriting under substitution.
It remains to show the variable condition
$\Var(v'\gamma\sigma_{\mte},Q) \subseteq \Var(\ell\sigma'\sigma,P)$.
Clearly
\begin{align*}
\TVar(\ell\sigma',\varphi,P)
&= \bigcup_{p \iN P} \TVar(\ell\sigma'|_p,\varphi)
= \bigcup_{p \iN P} \Var(\ell\sigma'|_p) \setminus \Var(\varphi) \\
&= \bigcup_{p \iN P} \Var(\ell\sigma'\sigma_{\mth}|_p)
= \Var(\ell\sigma'\sigma_{\mth},P)
\end{align*}
and similarly $\TVar(v',\psi,Q) = \Var(v'\gamma,Q)$. Hence
$\Var(v'\gamma,Q) \subseteq \Var(\ell\sigma'\sigma_{\mth},P)$
follows from
$\TVar(v',\psi,Q) \subseteq \TVar(\ell\sigma',\varphi,P)$.
Therefore,
\begin{align*}
\Var(v'\gamma\sigma_{\mte},Q) \subseteq
\Var(\ell\sigma'\sigma_{\mth}\sigma_{\mte},P) =
\Var(\ell\sigma'\sigma,P)
\end{align*}
as desired.
\qed
\end{proof}

\begin{theorem}
\label{thm:toy81LCTRStoy81}
If an \textup{LCTRS} $\xR$ is parallel closed then
$\ov{\xR}$ is parallel closed.
\end{theorem}

\begin{proof}
Let $\xR$ be a parallel closed LCTRS. First consider an arbitrary
critical pair $s \approx t \in \CP(\ov{\xR})$. From \Cref{thm:CP} we know
that there exist a constrained critical pair
$s' \approx t'~\CO{\varphi} \in \CCP(\xR)$ and a substitution $\sigma$
such that $s'\sigma = s$, $t'\sigma = t$ and $\sigma \vDash \varphi$.
Since the constrained critical pair is 1-parallel closed,
\Cref{lem:1pc-unconstrained} yields $s \pto \cdot \La[*] t$. Hence
$\ov{\xR}$ is 1-parallel closed.

Next consider an arbitrary parallel critical pair
$s \approx t \in \PCP(\ov{\xR})$.
\Cref{thm:PCP} yields a constrained parallel critical pair
$s' = \ell\sigma'[r_p\sigma']_{p \iN P} \approx r\sigma' = t'~\CO{\varphi}$
in $\CPCP(\xR)$ and a substitution $\sigma$ such that $s'\sigma = s$,
$t'\sigma = t$ and $\sigma \vDash \varphi$.
Since the constrained parallel critical pair is 2-parallel closed, by
\Cref{lem:2pc-unconstrained} there exist a term $v$ and a set of
parallel positions $Q$ such that $s \Ra[*] v \rptoa[Q] t$ and
$\Var(v,Q) \subseteq \Var(\ell\sigma'\sigma,P)$. Hence $\ov{\xR}$ is
2-parallel closed.
\qed
\end{proof}

Since left-linearity of $\xR$ is preserved in $\ov{\xR}$ and left-linear,
parallel closed TRSs are confluent by \Cref{thm:toy81}, we obtain the
following corollary via~\Cref{thm:PCP,thm:toy81LCTRStoy81}.
Again, $\xR$ only has to be left-linear in the variables
$x \notin \LVar$, since that is sufficient for $\ov{\xR}$ to be
left-linear.

\begin{corollary}
\label{cor:parallel-closed}
Every left-linear parallel closed \textup{LCTRS} is confluent.
\qed
\end{corollary}

We illustrate the corollary on a concrete example.

\begin{example}
\label{exa:toy81}
Consider the LCTRS $\xR$ over the theory \textsf{Ints} with the rules
\begin{align*}
\m{f}(\m{a}) &\R \m{g}(\m{4},\m{4}) &
\m{a} &\R \m{g}(\m{1} + \m{1},\m{3} + \m{1}) &
\m{g}(x,y) &\R \m{f}(\m{g}(z,y))~\CO{z = x - \m{2}}
\end{align*}
The constrained (parallel) critical pair
$\m{f}(\m{g}(\m{1} + \m{1},\m{3} + \m{1})) \approx
\m{g}(\m{4},\m{4})~\CO{\m{true}}$
originating from the peak
$\m{f}(\m{g}(\m{1} + \m{1},\m{3} + \m{1}))~\CO{\m{true}}
\rptoa[\{1\}] \m{f}(\m{a})~\CO{\m{true}}
\Rb[\epsilon] \m{g}(\m{4},\m{4})~\CO{\m{true}}$
is 2-parallel closed:
\begin{align*}
\m{f}(\m{g}(\m{1} + \m{1},\m{3} + \m{1})) \approx
\m{g}(\m{4},\m{4})~\CO{\m{true}}
&\sptoab[\geqslant 1]{\phantom{\{1\}}}
\m{f}(\m{g}(\m{2},\m{4})) \approx \m{g}(\m{4},\m{4})~\CO{\m{true}} \\
&\sptoab[\geqslant 2]{\{2\}}
\m{f}(\m{g}(\m{2},\m{4})) \approx \m{f}(\m{g}(\m{2},\m{4}))~\CO{\m{true}}
\end{align*}
Note that the condition
$\TVar(\m{f}(\m{g}(\m{2},\m{4})),\m{true},\{2\}) \subseteq
\TVar(\m{f}(\m{a}),\m{true},\{1\})$ is trivially satisfied. 
One easily checks that the corresponding constrained
critical pair is 1-parallel closed. Since the only other remaining
constrained critical pair is trivial, we conclude confluence by
\Cref{cor:parallel-closed}.
\end{example}

\section{Conclusion}

We presented a left-linearity preserving transformation from LCTRSs to
TRSs such that (parallel) critical pairs in the latter correspond to
constrained (parallel) critical pairs in the former. As a consequence,
confluence results for TRSs based on
restricted joinability conditions easily carry over to LCTRSs.
This was illustrated by generalizing the advanced confluence
results of van Oostrom~\cite{vO97} and Toyama~\cite{T81} from TRSs to
LCTRSs. We also proved that (local) confluence of terminating LCTRSs
over a decidable theory is undecidable in general.

\Cref{fig:venn-diagram} relates the confluence criteria in this
paper to the earlier ones from \cite{KN13,SM23}. The acronyms stand for
weak orthogonality (\textsf{WO}, \cite[Theorem~4]{KN13}),
strong closedness (\textsf{SC}, \cite[Theorem~2]{SM23}),
almost parallel closedness (\textsf{APC}, \cite[Theorem~4]{SM23}),
almost development closedness (\textsf{ADC},
\Cref{cor:almost-dev-closed}), and parallel closedness of (parallel)
critical pairs (\textsf{PCP}, \Cref{cor:parallel-closed}). All areas are
inhabited and the numbers refer to examples in this paper.

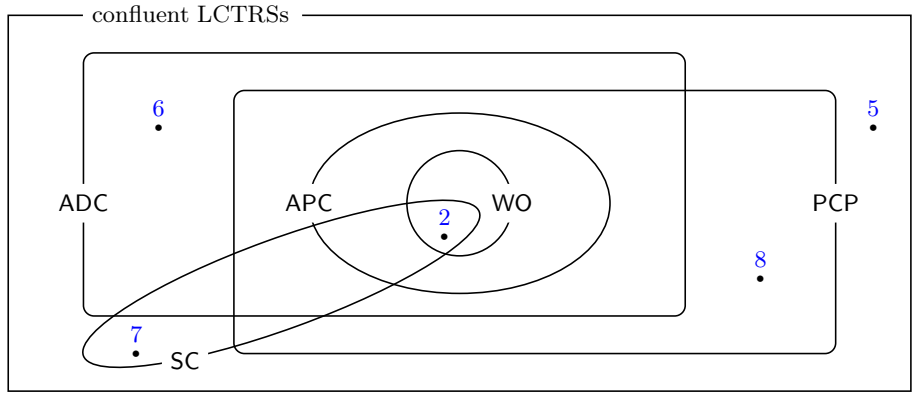
\begin{figure}
\centering
\begin{tikzpicture}
\draw[semithick] (0,0) rectangle (12cm,5cm);
\draw[fill=white,draw=white] (1cm,48mm) node[anchor=south west]
 {confluent LCTRSs} rectangle ++(29mm,5mm);
\draw[semithick] (60mm,25mm) ellipse (7mm and 7mm);
\draw[fill=white,draw=white] (65.5mm,22.5mm) rectangle ++(5mm,5mm);
\draw (67mm,25mm) node {\textsf{WO}};
\draw[semithick] (60mm,25mm) ellipse (20mm and 12mm);
\draw[fill=white,draw=white] (38.5mm,22.5mm) rectangle ++(5mm,5mm);
\draw (40mm,25mm) node {\textsf{APC}};
\draw[rounded corners,semithick] (10mm,10mm) rectangle ++(80mm,35mm);
\draw[fill=white,draw=white] (8.5mm,22.5mm) rectangle ++(5mm,5mm);
\draw (10mm,25mm) node {\textsf{ADC}};
\draw[rounded corners,semithick] (30mm,5mm) rectangle ++(80mm,35mm);
\draw[fill=white,draw=white] (108.5mm,22.5mm) rectangle ++(5mm,5mm);
\draw (110mm,25mm) node {\textsf{PCP}};
\draw[semithick,rotate=20] (39mm,1mm) ellipse (28mm and 6mm);
\draw[fill=white,draw=white] (20.5mm,2mm) rectangle ++(6mm,5mm);
\draw (23.5mm,4mm) node {\textsf{SC}};
\node (e7) at (10,1.5) {\tiny $\bullet$};
\node[yshift=1mm] at (e7.north) {\ref{exa:toy81}};
\node (e5) at (2,3.5) {\tiny $\bullet$};
\node[yshift=1mm] at (e5.north) {\ref{exa:almost-dev-closed}};
\node (e2) at (5.8,2.05) {\tiny $\bullet$};
\node[yshift=1mm] at (e2.north) {\ref{exa:weakly-orthogonal}};
\node (e6) at (1.7,0.5) {\tiny $\bullet$};
\node[yshift=1mm] at (e6.north) {\ref{exa:var-cond-equivalence}};
\node (e4) at (11.5,3.5) {\tiny $\bullet$};
\node[yshift=1mm] at (e4.north) {\ref{exa:trs-adc-notin-lctrs}};
\end{tikzpicture}
\caption{Relating confluence criteria for LCTRSs.}
\label{fig:venn-diagram}
\end{figure}

The confluence results of~\cite{KN13,SM23} have been implemented in
\crest.\footnote{\url{http://cl-informatik.uibk.ac.at/software/crest/}}
The tool is currently under heavy development, not only to
incorporate the results in this paper but also termination and completion
techniques. Confluence of LCTRSs is a new category in the upcoming
edition of the Confluence
Competition\footnote{\url{https://project-coco.uibk.ac.at/2024/}}
and we expect to present experimental results obtained with {\crest}
at the conference.

For TRSs numerous other confluence techniques, not based on restricted
joinability conditions of critical pairs, as well as sufficient
conditions for non-confluence are known~\cite{A13,HNvOO19,SH22,ZFM15}.
We plan to investigate which
techniques generalize to LCTRSs with our transformation.
The transformation also makes the formal verification of confluence
criteria for LCTRSs in a proof assistant a more realistic goal.

\begin{credits}
\subsubsection{\ackname}
The detailed feedback of the reviewers improved the presentation.
\subsubsection{\discintname}
The authors have no competing interests to declare that are relevant to
the content of this article.
\end{credits}

\bibliographystyle{splncs04}
\bibliography{references}

\end{document}